\documentclass[
superscriptaddress,
 amsmath,amssymb,
 aps,
pra,
twocolumn,
]{revtex4-2}
   
\usepackage{xcolor}
\usepackage{graphicx} 
\usepackage{dcolumn} 
\usepackage{bm} 
\usepackage{braket} 
\usepackage{amsthm}

\usepackage[ruled]{algorithm2e}

\newtheorem{theorem}{Theorem}
\newtheorem{corollary}{Corollary}
\newtheorem{lemma}{Lemma}

\newtheorem{definition}{Definition}

\newtheorem*{theorem-non}{Theorem}
\newtheorem*{lemma-non}{Lemma}
\newtheorem*{corollary-non}{Corollary}
\newtheorem*{proposition-non}{Proposition}

\usepackage{appendix}

\newcommand{\bx}{\bm{x}}
\newcommand{\by}{\bm{y}}
\newcommand{\bxi}{\bm{x}^{(i)}}
\newcommand{\bxik}{\bm{x}^{(i,k)}}
\newcommand{\rhoik}{\rho^{(i,k)}}

\newcommand{\byik}{\bm{y}^{(i,k)}}
\newcommand{\yik}{y^{(i,k)}}

\newcommand{\btheta}{\bm{\theta}}
\newcommand{\hath}{\hat{h}}

\newcommand{\DD}{\mathbb{D}}
\DeclareMathOperator{\RZ}{RZ}
\DeclareMathOperator{\RY}{RY}
\DeclareMathOperator{\CNOT}{CNOT}

\DeclareMathOperator{\RERM}{\mathsf{R}_{ERM}}
\DeclareMathOperator{\ROPT}{\mathsf{R}}
\DeclareMathOperator{\RGENE}{\mathsf{R}_{Gene}}

\DeclareMathOperator{\haar}{\text{Haar}}

\DeclareMathOperator{\Tr}{Tr}
\DeclareMathOperator{\Loss}{\mathcal{L}}

 \newcommand{\revise}[1]{\textcolor{teal}{#1}}

\usepackage{pifont}

\begin{document}

\title{Problem-Dependent Power of Quantum Neural Networks on Multi-Class Classification} 

\author{Yuxuan Du}
\email{duyuxuan123@gmail.com}
\affiliation{
 JD Explore Academy, Beijing 10010, China
}

\author{Yibo Yang}
\email{yibo.yang93@gmail.com}
\affiliation{
King Abdullah University of Science and Technology, Thuwal 4700, Kingdom of Saudi Arabia
}
\affiliation{
 JD Explore Academy, Beijing 10010, China
}
\author{Dacheng Tao}
\email{dacheng.tao@gmail.com}
\affiliation{Sydney AI Centre, School of Computer Science, The University of Sydney, NSW 2008, Australia}
\affiliation{
 JD Explore Academy, Beijing 10010, China
}
\author{Min-Hsiu Hsieh}
\email{min-hsiu.hsieh@foxconn.com}
\affiliation{
Hon Hai (Foxconn) Research Institute, Taipei, Taiwan
}

\begin{abstract}
Quantum neural networks (QNNs) have become an important tool for understanding the physical world, but their advantages and limitations are not fully understood. Some QNNs with specific encoding methods can be efficiently simulated by classical surrogates, while others with quantum memory may perform better than classical classifiers. Here we systematically investigate the problem-dependent power of quantum neural classifiers (QCs) on multi-class classification tasks. Through the analysis of expected risk, a measure that weighs the training loss and the generalization error of a classifier jointly, we identify two key findings: first, the training loss dominates the power rather than the generalization ability; second, QCs undergo a U-shaped risk curve, in contrast to the double-descent risk curve of deep neural classifiers. We also reveal the intrinsic connection between optimal QCs and the Helstrom bound and the equiangular tight frame. Using these findings, we propose a method that exploits loss dynamics of QCs to estimate the optimal hyper-parameter settings yielding the minimal risk. Numerical results demonstrate the effectiveness of our approach to explain the superiority of QCs over multilayer Perceptron on parity datasets and their limitations over convolutional neural networks on image datasets. Our work sheds light on the problem-dependent power of QNNs and offers a practical tool for evaluating their potential merit.
\end{abstract}

\maketitle

 \begin{figure*}
  \centering
\includegraphics[width=0.88\textwidth]{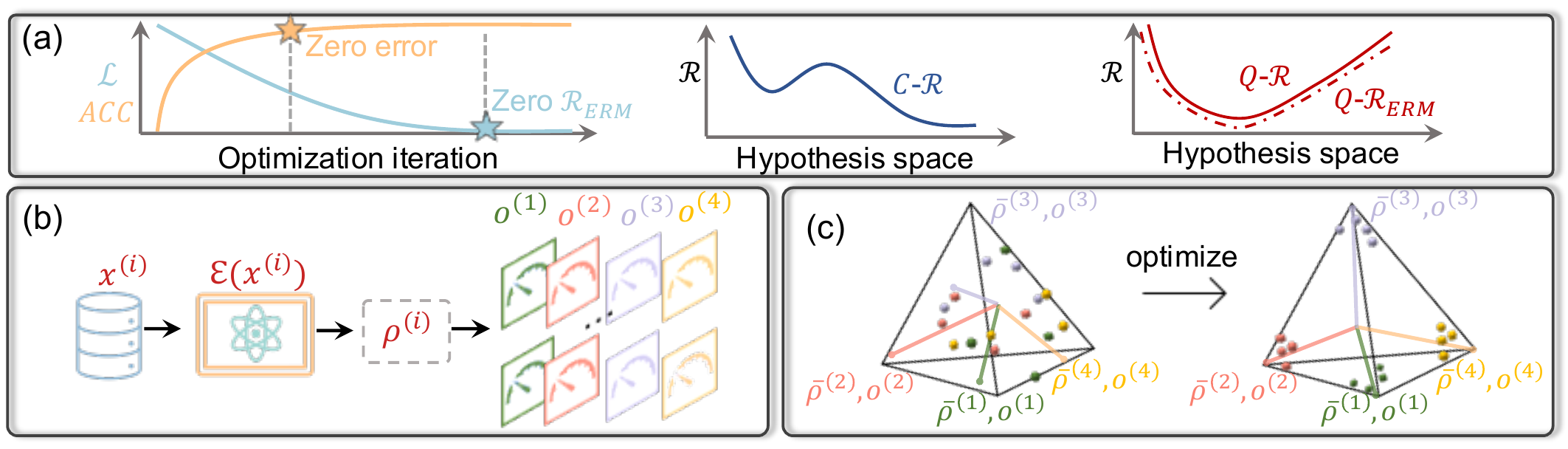}
\caption{\small{\textbf{Risk curve and geometry of the unified QCs.} (a) The left  plot shows the relation between the zero classification error (yellow star) and zero loss (blue star) during the training process. The middle and right plots  depict the risk curve of CCs and QCs (labeled by `C-$\mathcal{R}$' and `Q-$\mathcal{R}$'). Unlike CCs with a double-descent tendency, the risk curve of QCs yields a `U' shape whose power is dominated by the empirical risk (`Q-$\mathcal{R}_{\text{ERM}}$'), dominated by the expected risk. (b) The unified QC consists of two parts, the feature state $\rho$ and the measure operator $\bm{o}$. (c) Geometric interpretation of $\{\rhoik\}$ and $\bm{o}$ of QCs with (near) zero training loss: (i) the feature states associated with train samples belonging to the same class concentrate around their class-feature mean, i.e., $\bar{\rho}^{*(k)}:=\rho^{*(1,k)}=...= \rho^{*(n_c,k)}$ for $\forall k\in[K]$; (ii) the class-feature means are maximally distant with each other, i.e., $\Tr(\bar{\rho}^{*(k)} \bar{\rho}^{*(k')})\sim   \delta_{k,k'}$; (iii) the measure operator should align with class-feature means, i.e., $\Tr(\bar{\rho}^{*(k)} o^{*(k')})\sim  \delta_{k,k'}$.}}
\label{fig:scheme}
\end{figure*}

\noindent\textbf{\textit{Introduction}}.--- The advent of hardware fabrication pushes the boundary of quantum computing from verifying its superiority on artificial tasks \cite{arute2019quantum,zhong2020quantum,wu2021strong} to conquering realistic problems with merits \cite{mi2021information,Xia2021Quantum-enhanced,cerezo2022challenges}. This has led to the emergence of a popular paradigm known as quantum neural networks (QNNs), which combine variational quantum ansatzes with classical optimizers \cite{benedetti2019parameterized,cerezo2021variational}. So far, various QNN-based methods have been proposed to address difficult problems in areas such as quantum physics \cite{yuan2019theory, mcardle2020quantum,cirstoiu2020variational,google2020hartree}, quantum information theory \cite{romero2017quantum,du2021exploring,cerezo2020variational,bondarenko2020quantum}, combinatorial optimization \cite{farhi2016quantum,zhou2020quantum,harrigan2021quantum,zhou2022qaoa,pagano2020quantum}, and machine learning \cite{havlicek2018supervised,huang2021experimental,tian2022recent,wang2021towards,du2022theory}. Among these applications, QNNs are often deployed as \textit{quantum classifiers} (QCs) to predict correct labels of the input data \cite{schuld2019quantum,mitarai2018quantum,schuld2018circuit,li2021vsql,perez2020data,li2022recent}, e.g., categorize image objects \cite{du2021grover,chen2021end,peters2021machine}, classify  phases of quantum matters \cite{cong2019quantum,gong2022quantum,herrmann2022realizing,zhang2022experimental}, and distinguish  entangled states from separable states \cite{grant2018hierarchical,yin2022efficient}.

To comprehend the full potential of existing QCs and to spur the development of novel protocols, huge efforts have been made to unveil the learnability of QCs  \cite{abbas2020power,du2021learnability,huang2021power}. Prior literature establishes the foundations of QCs from three primary aspects, i.e., model capacity \cite{du2018expressive,haug2021capacity,shen2020information,wu2021expressivity}, trainability \cite{anschuetz2022beyond,shirai2021quantum,holmes2021connecting}, and generalization  \cite{banchi2021generalization,caro2021encoding,caro2021generalization,du2022efficient,gyurik2021structural,peters2022generalization,huang2021information}. Nevertheless, the power and limitation of QCs have rarely been proven  \cite{huang2022quantum,huang2021information,ciliberto2020statistical,Landman2022Classically,de2022limitations,schreiber2022classical}. Meanwhile, previous results cannot rigorously explain the empirical observations such that QCs generally exhibit superior performance on handcraft or quantum data \cite{huang2021power,liu2021rigorous} but inferior performance on realistic problems \cite{qian2021dilemma}. As a result, the specific problem domains in which QCs can excel effectively are still uncertain.

A principal criteria in characterizing the power of a classifier is the expected risk \cite{mohri2018foundations}, which weighs the empirical risk (training loss) and the generalization error (test loss) jointly. An \textit{optimal} classifier is one which achieves zero expected risk \footnote{We emphasize that the assessment of classifiers' power is not limited to a single measure. In general, the classification accuracy serves as the primary criterion, while other factors are also critical to consider, including robustness in domains such as healthcare and finance. Therefore, the optimality of a classifier is contingent upon the specific problem domain. Refer to SM~A for further details.}. As shown in Fig.~\ref{fig:scheme}(a), the success of deep neural classifiers is attributed to their double-descent risk curves \cite{Nakkiran2020Deep,belkin2019reconciling}. This means that as the hypothesis space is continually expanded, the expected risk of a trained deep neural classifier initially decreases, increases, and when it overfits the train set, undergoes a second descent. Given the fundamental importance of the expected risk, it demands to distill ubiquitous rules that capture the risk curve of diverse QCs.

In this study, we unify a broad class of QCs in the same framework and understand their problem-dependent ability under the expected risk (see Fig.~\ref{fig:scheme}(b)). Our analysis reveals two substantial outcomes: (i) trainability dominates QCs’ ability more than generalization ability; (ii) QCs undergo a U-shaped risk curve instead of the double-descent curve for CCs  with respect to the increased hypothesis space. These outcomes consolidate and refine previous observations. The first outcome suggests that the deficiency of QCs on classical data may stem from their limited ability to fit the train set, leading high classification error and training loss.  The second outcome highlights the distinct learning behavior of QCs and CCs. Despite the fact that over-parameterization is crucial to enhance the performance of CCs, it may adversely affect the power of QCs. In line with the diverse dynamics of the risk curves for QCs and CCs, we devise an efficient problem-dependent method to determine the suitable circuit depth of a QC with a near-optimal risk. Numerical simulations validate our theoretical results.

\medskip
\noindent \textit{\textbf{Expected risk}}.--- Let us first introduce a $K$-class ($K \geq 2$) classification task. Denote the input space as $\mathcal{X}$, the label (class) space as $\mathcal{Y} = \{1, \cdots, K\}$, and the train set as $\mathcal{D}=\bigcup_{k=1}^K \{(\bxik, \yik)\}_{i=1}^{n_k}$ with $|\mathcal{D}|$ samples  drawn i.i.d. from an unknown probability distribution $\mathbb{D}$ on $\mathcal{Z}=\mathcal{X}\times \mathcal{Y}$. In standard scenarios, the number of train samples in each class is the same, i.e., $n_1=...=n_k\equiv n_c$ and $|\mathcal{D}|:=n=Kn_c$. The purpose of a classification algorithm $\mathcal{A}$ is using $\mathcal{D}$ to infer a hypothesis (a.k.a., a classifier) $h_{\mathcal{A}_\mathcal{D}}:\mathcal{X} \rightarrow \mathbb{R}^K$ from the hypothesis space $\mathcal{H}$ to separate train examples from different classes. This is equivalent to identifying an optimal hypothesis in $\mathcal{H}$ minimizing the \textit{expected risk} $\ROPT (h)=\mathbb{E}_{(\bx, \by)\sim \mathbb{D}}[\ell(h(\bx), y)]$, where $\ell(\cdot, \cdot)$ is the per-sample loss and for clarity we specify it as the square error with $\ell(\bm{a}, \bm{b}) =\frac{1}{2} \|\bm{a}- \bm{b}\|_2^2$ \cite{bishop2006pattern}. Unfortunately, the inaccessible distribution $\DD$ forbids us to assess the expected risk directly. In practice, $\mathcal{A}$ alternatively learns an \textit{empirical classifier} $\hath\in \mathcal{H}$, as the global minimizer of the (regularized) loss function 
 \begin{equation}\label{eqn:gene_obj_func}
  \mathcal{L}(h, \mathcal{D}) =  \frac{1}{n}\sum_{i=1}^{n_c}\sum_{k=1}^K \ell(h(\bxik), \yik)  + \mathfrak{E}(h),
\end{equation}
where $\mathfrak{E}(h)$ is an optional regularizer.   
 
The foremost role of the risk suggests that it can serve as a critical measure to assess the power of QCs. Unlike conventions merely focusing on a QC on one specific task, what we intend to is unearthing \textit{ubiquitous rules} of QCs encompassing  diverse constructions and tasks. To reconcile the intractability of  $\ROPT (\hath)$ and proceed further  analysis, we decompose it into two measurable terms,  
\begin{equation}\label{Eqn:opt_risk}
  \ROPT(\hath)  =   \mathsf{R}_{\text{ERM}}(\hath) + \RGENE(\hath),   
 \end{equation}
 where $\RERM(\hath)=\frac{1}{n}\sum_{i=1}^n\sum_{k=1}^K \ell(\hath(\bxik), \yik)$ is  the \textit{empirical risk} and  $\mathsf{R}_{\text{Gene}}(\hath)= \ROPT(\hath) - \mathsf{R}_{\text{ERM}}(\hath)$ is the \textit{generalization error}. Based on Eq.~(\ref{Eqn:opt_risk}),  quantifying the optimality of QCs amounts  to deriving under what conditions do QCs commit both the vanished $\RERM$ and $\RGENE$. We note that the vanishing conditions serve as sufficient criteria for the success of QCs. In practical scenarios, achieving $\mathsf{R}_{\text{ERM}}(\hath)\rightarrow 0$ may be challenging and not necessary when classification accuracy is the sole concern. Nonetheless, considering multiple metrics to assess classifier power, $\mathsf{R}_{\text{ERM}}(\hath)\rightarrow 0$ becomes important (see SM~A for  explanations \footnote{See Supplemental Material (SM) for the proofs, the implications of Theorem 1 and Corollary 1, and the omitted details of algorithmic implementation and numerical simulations.}).

To better elucidate our results, let us recall that the general form of  QC is $\hath_Q= \arg\min_{h_Q\in \mathcal{H}_Q} \mathcal{L}(h_Q, \mathcal{D})$, where $\mathcal{L}$ is defined in Eq.~(\ref{eqn:gene_obj_func}) and $\mathcal{H}_Q$ is the hypothesis space. For an $N$-qubit QC, its hypothesis space is 
\begin{equation}\label{eqn:hypo-QC-convention}
   \mathcal{H}_Q = \left\{\left[h_Q(\cdot, U(\btheta), O^{(k)})\right]_{k=1:K}\Big|\btheta \in \bm{\Theta}\right\}, 
\end{equation}   
where $[\cdot]_{k=1:K}$ is a $K$-dimensional vector, its $k$-th entry $h_Q(\bm{x}, U(\btheta), O^{(k)})=\Tr(O^{(k)}U(\btheta)\sigma(\bx) U(\btheta)^{\dagger})$ for $\forall k\in[K]$ refers to the output (prediction) of quantum circuits, $\sigma(\bx)=U_E(\bx)(\ket{0}\bra{0})^{\otimes N}U_E(\bx)^{\dagger}$ is the input state of $\bx$ with the encoding circuit $U_E(\cdot)$, $\bm{O}=\{O^{(k)}\}_{k=1}^K$ is a set of  measure operators, and $U(\btheta)$ is the adopted ansatz with  trainable parameters $\btheta$ living in the parameter space $\bm{\Theta}$. Without loss of generality, we define  $U(\btheta)=\prod_{l=1}^{N_t}(u_l(\btheta)u_e) \in\mathcal{U}(2^N)$, where $u_l(\btheta)\in\mathcal{U}(2^m)$ is the $l$-th parameterized quantum gate operated with at most $m$ qubits ($m\leq N$) and $u_e$ refers to  fixed quantum gates. Similarly, we define $U_E(\bx)=\prod_{g=1}^{N_g}u_g(\bx) \in\mathcal{U}(2^N)$, where $u_g(\bx)\in\mathcal{U}(2^m)$ refers to the $g$-th  quantum gate operated with at most $m$ qubits, and $N_g$ gates contain $N_{ge}$ tunable gates and $N_g-N_{ge}$ fixed gates.  

 Due to the diverse constructions of $U(\btheta)$ and $U_E(\cdot)$, it necessitates to unify various QCs into the same framework to obtain the generic results. Notably, the unified QC should be \textit{agnostic to} particular forms of these two terms and capable of addressing both the under- and over-parameterized regimes. Note that the definition of over-parameterization varies in the literature when considering $\RERM$ and $\RGENE$ (see SM B for details). For this reason, we define over-parameterization as follows. 
\begin{definition}[Over-parameterization of QCs]\label{def:over-para} 
The over-parameterization regime of QCs is 
	$N_t>n$ in terms of generalization and when $U(\btheta)$ forms a 2-design in terms of trainability.
\end{definition}
 
\noindent To satisfy the above requirements, a feasible way is rewritten $h_Q(\cdot, U(\btheta), O^{(k)})$ as  
\begin{equation}\label{eqn:hypo-qc-uni}
h_Q(\bxik, U(\btheta), O^{(k)}) :=  \Tr(\rhoik o^{(k)}),~\forall k\in[K],
\end{equation}
where $O^{(k)}=\mathbb{I}_{2^{N-D}}\otimes o^{(k)}$ with the nontrivial local operator $o^{(k)}\in \mathbb{C}^{2^{D}\times 2^{D}}$, $D$ describes the locality with $2^D\geq K$, and $\rho^{(i,k)}= \Tr_D(U(\btheta)\sigma(\bxik)U(\btheta)^{\dagger})$ corresponds to the state before measurements, named as \textit{feature state}. See Fig.~\ref{fig:scheme}(b) for an intuition.

We now exploit the unified framework to analyze the expected risk of QCs. Let $\bm{\rho}=\{\rhoik\}$ and $\bm{o}=\{o^{(k)}\}$ be two sets collecting all feature states and measure operators. The following theorem exhibits properties of $\bm{\rho}$ and   $\bm{o}$ in which QCs achieve a low expected risk, where the formal statement and the proof are deferred to SM~C. 
 \begin{theorem}[informal]\label{thm:opt_learnability_QC}
 Following notations in Eqs.~(\ref{eqn:gene_obj_func})-(\ref{eqn:hypo-qc-uni}), the global optimizer $(\bm{\rho}^*,\bm{o}^*)$ that can reach $\RERM\rightarrow 0 $  satisfies the following properties: (i) the feature states have the vanished variability in the same class; (ii) all feature states are equal length and are orthogonal in the varied classes; (iii) any feature state is alignment with the measure operator in the same class. Moreover, when the  train data size is $n \gg O(KN_{ge}\log \frac{KN_g}{\epsilon\delta})$ with $\epsilon$ being the tolerable error, with probability $1-\delta$, the expected risk of this optimal QC tends to be zero, i.e., $\ROPT(\hath_Q) \rightarrow 0$. 
\end{theorem}

\noindent Conditions (i)-(iii) visualized in Fig.~\ref{fig:scheme}(c) sculpt the geometric properties of $\bm{\rho}^*$ and $\bm{o}^*$ achieving $\RERM\rightarrow 0$. The mean feature states of each class $\{\bar{\rho}^{*(k)}\}$ compose the orthogonal frame and Condition (iii) suggests that the optimal measure operators $\{o^{*(k)}\}$ also satisfy this orthogonal frame \footnote{For any class $k\in [K]$, the mean feature state is defined by $\bar{\rho}^{*(k)}= (\rho^{*(1,k)} +\rho^{*(2,k)}+...+ \rho^{*(n_c,k)})/n_c$}. Since any orthogonal frame can trivially be turned into a general simplex equiangular tight frame (ETF) \cite{papyan2020prevalence} by reducing its global mean, it can be concluded that $\{\bar{\rho}^{*(k)}\}$ or $\{o^{*(k)}\}$ forms the general simplex ETF. Note that when $2^D=K$, the orthogonal frame is equivalent to a formal ETF \cite{sustik2007existence}. Building on the extensive research on general simplex and formal ETFs in deep learning and quantum information, we subsequently explore the intrinsic connection between QCs and deep neural classifiers and study the power of QCs in the view of quantum state discrimination (refer to SM~H for the omitted definitions and explanations). 

In the context of deep learning, Refs.~\cite{papyan2020prevalence,liu2023inducing,yang2022we} proved that for a deep neural classifier with zero training loss, its last-layer features also form a general simplex ETF, dubbed neural collapse. In this regard, QCs and deep neural classifiers experience  similar learning behaviors, in which the corresponding features tend to form a general simplex ETF to reach zero training loss.

We next understand the results of Theorem \ref{thm:nogonc} from the perspective of quantum state discrimination \cite{bae2015quantum}. The setting $2^D\geq K$ in Eq.~(\ref{eqn:hypo-qc-uni}) ensures that the global optimizer $(\bm{\rho}^*,\bm{o}^*)$ in Theorem \ref{thm:nogonc} maximizes the Helstrom bound, i.e., for any two varied classes, $o^{(k)}$ and $o^{(k')}$ classify $\bar{\rho}^{*(k)}$ and  $\bar{\rho}^{*(k')}$ with probability $1$. This observation explains the ultimate limit of QCs observed in \cite{zhang2021fast}. Intriguingly, Ref.~\cite{banchi2021generalization} achieved the similar geometric properties in the view of information theory.

The maximized Helstrom bound when $2^D\geq  K$ hints that the locality $D$ of $\{o^{*(k)}\}$ should be carefully selected. Particularly, although the construction of $\{o^{*(k)}\}$ is flexible in QCs, a large $D$ may incur the barren plateaus  \cite{cerezo2020cost,sack2022avoiding}. To this end, it is interesting to explore the properties of QCs when $2^D<K$. In this case, achieving zero error probability in discriminating different feature states is unattainable, supported by the results of \cite{montanaro2008lower}. Moreover, the feature states of optimal QCs with $\RERM\rightarrow 0$ tend to form formal ETFs rather than general simplex ETFs, approaching the corresponding lower bound. Notably, unlike general simplex ETF always exists for any $D$, formal ETFs arise for very few pairs $(2^D, K)$ \cite{sustik2007existence}. A possible solution is using symmetric informationally complete   POVM to attain the lower bound of the error probability  \cite{renes2004symmetric,scott2006tight,garcia2021learning}, since it is a special case of formal ETF with $2^D=\sqrt{K}$.

On the technical side, we prove that the generalization error  of QCs is upper bounded by
\begin{equation}\label{eqn:risk_curve}
 \RGENE(\hath_Q) \leq  \tilde{O}(K\epsilon + \sqrt{{g(\RERM(\hath_Q), n, K)}/{n}}),   
\end{equation}  
where $g(\cdot)$ decreases from $n$ to $K$ when $\RERM(\hath_Q)\rightarrow 0$ and $\epsilon$ is the tolerable error. Connected with  Eq.~(\ref{Eqn:opt_risk}), we prove $\ROPT(\hath_Q) \rightarrow 0$ by separately showing that Conditions (i)-(iii) lead to $\RERM\rightarrow 0$ and $n \gg O(KN_{ge}\log \frac{KN_g}{\epsilon\delta})$ warrants $\RGENE\rightarrow 0$. The derived bound for $\RGENE$  surpasses prior results because it is the first non-vacuous bound in the over-parameterized regime  of Definition \ref{def:over-para}. Namely, previous generalization bounds are algorithm-independent and reflect the influence of expressivity \cite{caro2021generalization,du2022efficient,gyurik2021structural}, which causes $\RGENE\leq O(N_t/n)$. Accordingly, these bounds are vacuous in the over-parameterized regime with $N_t>n$. By contrast, our bound ensures a non-vanishing generalization error even $N_t>n$, since it is algorithmic-dependent and not explicitly relying on $N_t$.

According to above analysis, the challenges in satisfying Conditions (i)-(iii) and the well controlled generalization error pinpoint that the risk of a QC is mostly dominated by its empirical loss. As such, the core in devising QCs is tailoring $\mathcal{H}_Q$ and adopt advanced optimization techniques so that $\hath_Q$ can fulfill Conditions (i)-(iii).  For example, when $K$ is large, Pauli-based measurements are preferable, which allows classical shadow  techniques to accelerate the training procedure \cite{huang2020predicting,huang2022learning}.

\noindent \textit{Remark}. Although the results related to the zero training loss improve the interpretability of QCs, exact satisfaction of this condition in the realistic scenario may be difficult and unrealistic in practice. In SM~H, we discuss the expected risk of QCs under the approximate satisfaction.

 \begin{figure}
  \centering
  \includegraphics[width=0.49\textwidth]{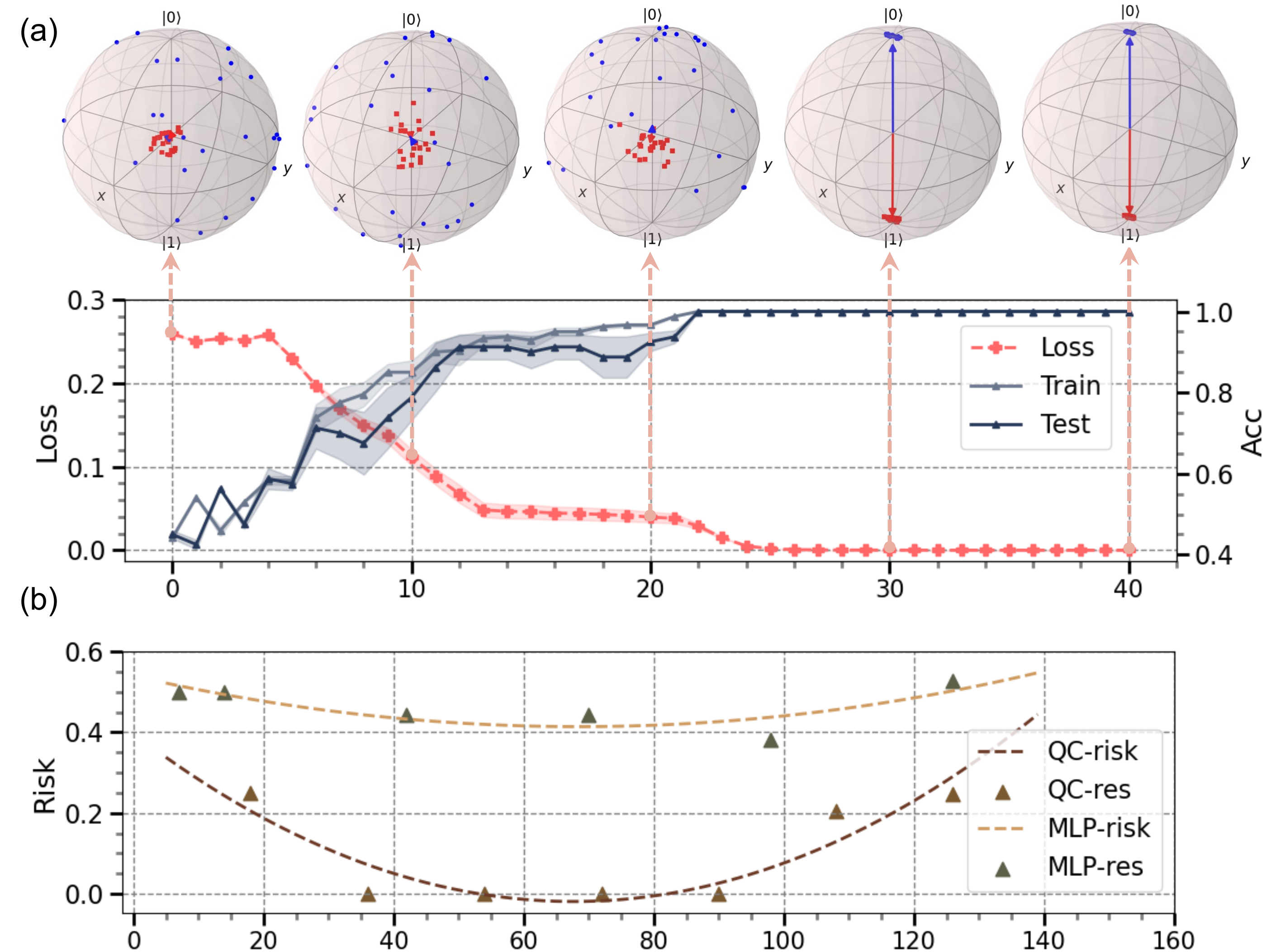}
  \caption{\small{\textbf{Binary classification on the parity dataset.} (a) The learning dynamic of QC. The $x$-axis denotes the epoch numbers. Shaded region represents variance. The Bloch spheres display the quantum feature states at different epochs.  (b) The fitted risk curve of QC and MLP. The $x$-axis denotes the number of trainable parameters. The label `\textit{QC-risk}' (`\textit{MLP-risk}') refers to the fitted risk curve of QC and MLP. The label `\textit{QC-res}' (`\textit{MLP-res}') refers to the collected results used for fitting the curves.}}
  \label{fig:sim-synthetic}
 \end{figure}

\medskip
\noindent\textit{\textbf{U-shaped risk curve.}}--- The risk curve concerns how the risk of a classifier behaves with the varied hypothesis space. It is desired that as with CCs, QCs follow a double-descent risk curve. If so, over-parameterization in Definition \ref{def:over-para} could serve as a golden law in QCs' design \footnote{As previously mentioned, our primary focus lies on the expected risk, which serves as a sufficient condition for the success of QCs. In scenarios where classification is the sole determinant of power, over-parameterized QCs may also achieve perfect classification accuracy.}. However, the corollary below refutes this conjecture.  
\begin{corollary}\label{thm:nogonc}
Following notations in Theorem \ref{thm:opt_learnability_QC}, when $\{U_E(\bx)|\bx\in \mathcal{X}\}$  forms a 2-design, with probability $1-\delta$, the empirical QC follows  $|	\Tr\left(\sigma(\bxik)\sigma(\bx)\right) - {1}/{2^N}| \leq \sqrt{{3}/{(2^{2N}\delta)}}$. When $\{U(\btheta)|\btheta\in \Theta\}$  forms a 2-design, with probability $1-\delta$, the empirical QC follows $|\Tr(\rhoik o^{(k')}) -\frac{\Tr(o^{(k')})}{2^{D}} |  < \sqrt{\frac{\Tr(o^{(k')})^2 +  2\Tr((o^{(k')})^2)}{2^{2D} \delta }}$.
\end{corollary}
\noindent The proof is deferred to SM~B. The achieved results reveal the caveat of deep QCs. When $U_E(\bx)$ is deep, two encoded states $\sigma(\bxik)$ and $\sigma(\bm{x}^{(i',k)})$ from the same class tend to be orthogonal. Besides, QC's output with deep $U(\btheta)$ concentrates to zero, regardless how $o^{(k')}$ and $\rhoik$ are selected. This violates Condition~(iii). Overall, in conjunction with Eqs.~(\ref{Eqn:opt_risk}) \& (\ref{eqn:risk_curve}), over-parameterization increases $\RERM$  and thus $\ROPT$, which suggests that QCs experience a \textit{U-shaped risk curve}. This phenomenon aligns with the observations made in Ref.~\cite{banchi2021generalization}, where the presence of a U-shaped curve in QCs was uncovered using information theory tools. The U-shaped curve hints the varied design strategies for QCs and variational quantum Eigensolvers, since the latter can benefit from over-parameterization    \cite{liu2022laziness,liu2022analytic,wang2022symmetric,you2022convergence}. Moreover, when the employed $U(\btheta)$ forms a 2-design, QCs not only experience barren plateaus during the training phase but also flatten the entire loss landscape, leading that the global minima of loss function must be greater than zero. Therefore, the rule of thumb in QCs' construction is slimming $\mathcal{H}_Q$ to find the valley region, which echoes with  quantum metric learning and quantum self-supervised learning \cite{lloyd2020quantum,nghiem2021unified,larose2020robust,jaderberg2022quantum,Yang2022Analog}.  

\medskip
\noindent\textbf{\textit{Probe power of QCs via loss dynamics.}}--- The distinct tendency of the risk curves between QCs and CCs indicates a distinct way to improve their performance. As shown in Fig.~\ref{fig:scheme}(a), given a specific dataset and a specified ansatz, the evolution of the risk curve with the hypothesis space is dominated by the number of parameters $N_t$. In other words, it is desirable to find an optimal $N_t$ whose expected risk is lower than other settings. The proved learning behavior of QCs allows us to effectively fit their risk curve according to the loss dynamics and estimate a near-optimal $N_t$ whose  expected risk is around the basin.   Specifically, our method contains three steps. First, $W$ tuples of $\{n, N_t, T\}$ are initialized so that the collected risk points of QC span the basin area. Second, we execute QC  under these $W$ hyper-parameter settings and fit the loss dynamics to attain the risk curve. Last, we use the fitted risk curve to estimate $N_t$ corresponding to the basin. Note that the proposed method is complementary to the recent results in geometric QNN. Refer to See SM~I for elaborations.

\medskip
\noindent\textbf{\textit{Numerical results}.}--- We conduct numerical simulations to exhibit that the power of QCs on different classification tasks can be interpreted by the derived risk curve and feature states.  The omitted construction details and results are deferred to SM~G.

 \begin{figure}[!t]
	\centering
\includegraphics[width=0.49\textwidth]{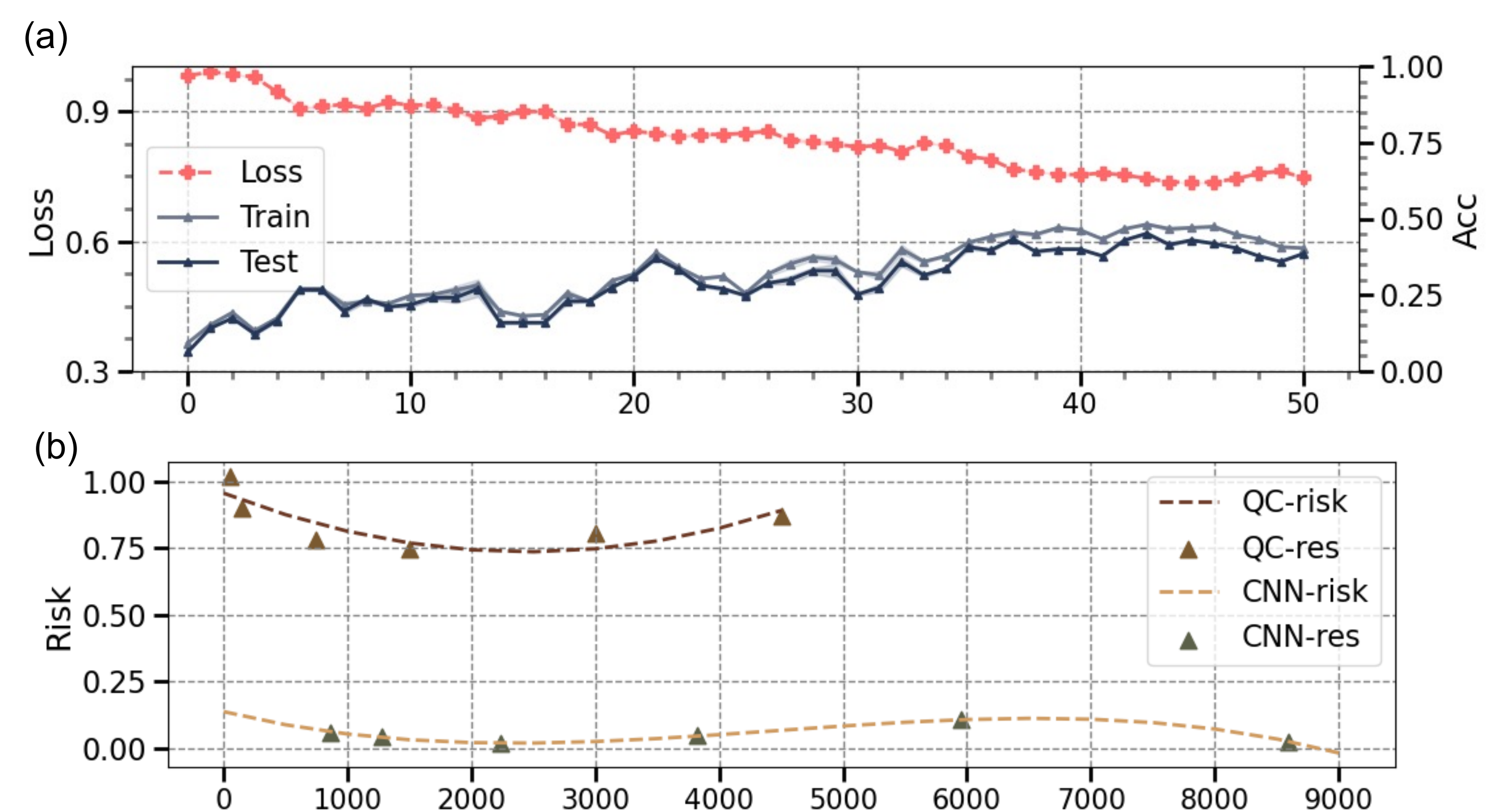}
	 \caption{\small{\textbf{Multi-class classification on the image dataset with $K=9$.} (a) The learning performance of QC when the layer number is $50$.    (b) The fitted risk curve of QC and CNN. All labels have the same meaning with those used in Fig.~\ref{fig:sim-synthetic}.}}
  \label{fig:sim-FMNIST}
\end{figure}

We first apply QC to accomplish the binary classification on the parity dataset \cite{cross2015quantum,riste2017demonstration,sen2022variational}. The number of qubits is $N=6$ and the hardware-efficient ansatz is adopted to realize $U(\btheta)$. The gradient descent method is used to update $\btheta$. Two measure operators are $o^{(1)}=\ket{0}\bra{0}$ and $o^{(2)}=\ket{1}\bra{1}$. The simulation results of QC with $N_t=54$ are displayed in Fig.~\ref{fig:sim-synthetic}(a). The averaged train (test) accuracy steadily grows from $44.1\%$ to $100\%$  within $22$ epochs, and the corresponding loss decreases from $0.26$ to $4\times 10^{-5}$. The dynamics of the feature states $\{\rho^{(i, t)}\}$ visualized by Bloch spheres echo with the theoretical analysis. Besides, QC becomes more robust when we continue the training \footnote{Although the train (test) accuracy reaches the optimum, the loss can be further reduced and suggests a lower risk warranted by Theorem \ref{thm:opt_learnability_QC}.}. We further compare the risk curve between QC and  multilayer Perceptron (MLP) by fitting their risk curves following the proposed method. As shown in Fig.~\ref{fig:sim-synthetic}(b), QC clearly outperforms MLP when $N_t \in [20, 140]$ and its valley is around $N_t=70$ \footnote{The disappeared  double-descent curve of MLP is caused by the limited train data. In other words, over-parameterization and sufficient train data are two necessary conditions to induce the double-descent curve, while parity dataset can only provide limited train data.}.

We then apply QC to learn the Fashion-MNIST image dataset with $K=9$ \cite{xiao2017online}. The employed number of qubits is $N=10$ and the Pauli-based measure operators are employed. Convolutional neural networks (CNN) is exploited as the reference. For all classifiers, the number of epochs is fixed to be $T=50$ and $N_t \in [60, 9000]$. Each setting is repeated with $3$ times. As shown in Fig.~\ref{fig:sim-FMNIST}, when $N_t=1500$, both train and test accuracies of QC are about $50\%$, which is inferior to CNN under the similar setting. To explore the potentials of QC, we compare their risk curves. As shown in Fig.~\ref{fig:sim-FMNIST}(b), unlike the parity dataset, QC is evidently inferior to CNN on Fashion-MNIST dataset.

\medskip 
\noindent \textit{\textbf{Discussions.}}--- We understand the power of QCs in terms of the expected risk and exhibit that the efficacy of QCs is dependent on the problem at hand. Leveraging  the derived U-shaped risk curve, we present a concise technique to enhance the performance of a given quantum classifier by fitting its loss dynamics. There are several interesting future research directions. First,   is it necessary to redesign QCs such as nonlinear QCs \cite{schuld2022quantum,holmes2021nonlinear} that can also exhibit a double-descent risk curve? Second, it is intriguing to extend the developed non-vacuous generalization error bound of QCs to other scenarios to identify potential quantum advantages.  
 
 \begin{acknowledgements}
The authors thank Xinbiao Wang for valuable input and inspiring discussions.
\end{acknowledgements}

\onecolumngrid

\appendix 

\renewcommand{\appendixname}{SM}

 \renewcommand\thefigure{\thesection.\arabic{figure}}    
\bigskip
The organization of the supplementary materials (SM) is as follows. In SM~\ref{append:risk-adv}, we discuss why the expected risk is an appropriate measure to comprehend the power of quantum classifiers. In SM~\ref{sec:preliminary}, we present definitions of over-parameterization in classical and quantum machine learning.  Subsequently, in SM \ref{append:thm:opt_learnability_QC-thm3}, we show the results related to the proof of Theorem 1. Two core lemmas used in the proof of Theorem 1 are demonstrated in SM \ref{append:sec:pro1}  and SM \ref{append:sec:proof-thm1}, respectively. Next, in SM \ref{append:sec:gene_QNN}, we exhibit the proof of Lemma C3. Then, we provide the proof of Corollary 1 in SM~\ref{append:proof-them4}. More details for the implications of Theorem 1 are elucidated in SM~\ref{append:sec:imp-thm1}.  In SM \ref{append:sec:alg-imp}, we elaborate on the proposed method to probe and enhance the power of quantum classifiers. In the end, we illustrate  the details of numerical simulations in SM \ref{append:sec:sim-res}.

\tableofcontents

\section{Different measures of the power of quantum machine learning models}\label{append:risk-adv}

In the main text, we investigate the power of quantum classifiers (QCs) by focusing on the scenario where the expected risk tends to zero \cite{mohri2018foundations}.  However, it is noteworthy    that the power of classifiers cannot be captured by a single measure alone \cite{schuld2022quantum}. Various measures, including asymptotic runtime \cite{harrow2009quantum}, sample complexity \cite{arunachalam2017guest}, and expected risk, contribute to our understanding of the power of QCs. In particular, the expected risk plays a fundamental role and has been leveraged to assess the performance of QCs in previous research. For instance, Ref.~\cite{sharma2022reformulation} presents a quantum version of the no-free-lunch theorem, demonstrating that entangling the input state with a reference system can lead to a lower expected risk compared to using non-entangled input states in unitary learning tasks. Ref.~\cite{huang2021information} shows that classical and quantum learning models exhibit similar performance, as measured by the expected risk, when predicting outcomes of physical experiments.

While the expected risk is of great importance, it is worth noting that achieving a vanished expected risk is a sufficient condition rather than a necessary one for the success of QCs. In practice, both QCs and deep neural classifiers can achieve perfect training classification accuracy even when their empirical risk is non-zero. Therefore, if classification accuracy is the sole measure of the power of QCs, a zero empirical risk is not necessarily required.

However, when considering measures of power beyond classification accuracy, the vanished expected risk becomes significant from both theoretical and practical perspectives. In the context of adversarial learning \cite{goodfeloow2014explain}, the robustness of the classifier is a crucial concern in the sense that a trained QC is expected to to maintain its predictions even when the input is slightly perturbed by an attacker. Our results demonstrate that continuously optimizing QC from perfect classification to perfect training enhances their adversarial robustness by maximizing the distance between examples with different labels. This finding aligns with the field of quantum adversarial learning \cite{lu2020quantum,ren2022experimental} and provides valuable insights into the practical utility of QCs in domains where robustness is critical, such as finance and healthcare.

Furthermore, empirical studies have shown that further improvement in test accuracy can be achieved by continuing the optimization of the classifier after achieving perfect training accuracy \cite{papyan2020prevalence,belkin2019reconciling,xu2023dynamics}. Halting optimization once the training accuracy is perfect may result in suboptimal test accuracy, while continuing optimization can be time-consuming. Thus, understanding the number of essential optimization steps is crucial for designing more efficient optimization methods that yield classifiers with good test accuracy and reduced runtime costs. The results presented in Theorem 1 partially address this knowledge gap by shedding light on the consequences of continuously training quantum classifiers even after reaching zero training classification error. We acknowledge that these results may establish on unrealistic assumptions,  whereas they improve the interpretability of QCs.

\section{More explanations of over-parameterization in classical and quantum machine learning}\label{sec:preliminary}
 
	In this section, we provide a recapitulation of the various definitions of over-parameterization that are used in classical and quantum machine learning communities. Subsequently, we elucidate how these definitions correlate with the over-parameterization of Definition 1  utilized in this work.

The varied definitions of over-parameterization are summarized in Table~\ref{tab:over-para-sum}. Particularly, in the context of deep learning, over-parameterization typically refers to the scenario where the number of trainable parameters, denoted as $N_t$, is significantly larger than the number of training samples, denoted as $n$, i.e., $N_t \gg n$ \cite{zhang2021understanding,allen2019learning,allen2019convergence}. However, in the realm of quantum machine learning, the definition of `over-parameterization' is varied depending on whether it is applied to optimization or learning tasks. In optimization tasks, such as estimating the ground energy of an $N$-qubit Hamiltonian $H$, over-parameterization of quantum neural networks (QNNs) may be defined as $N_t \sim O(\exp(N))$ for problem-agnostic ansatz  \cite{liu2023analytic}, or as $N_t \sim \Omega(d_{\text{eff}})$ for problem-informed ansatz \cite{larocca2021theory,you2022convergence,wang2022symmetric,sauvage2022building}, where $d_{\text{eff}}$ is the effective dimension and could be $d_{\text{eff}} \sim O(poly(N))$ for Hamiltonians with favorable  symmetric properties. In learning tasks, there are two different versions of `over-parameterization' in the literature. In Ref.~\cite{you2023analyzing}, the condition of over-parameterization is defined as $N_t \rightarrow \infty$ to analyze the convergence of QNNs in classification tasks through the lens of quantum neural tangent kernel. On the other hand, when prior information of the explored dataset is exploited to design QNNs, also known as geometric QNNs \cite{larocca2022group,meyer2023exploiting,ragone2022representation,nguyen2022theory}, Ref.~\cite{schatzki2022theoretical} demonstrates that the over-parameterization condition is $N_t \sim O(N^3)$ with $N$ being the number of qubits in QNNs. Notably, Ref.~\cite{schatzki2022theoretical} shows that the generalization error bound of geometric QNNs is $O(N^3/n)$, which becomes vacuous when $N^3>n$ (see SM.~\ref{append:sec:imp-thm1} for explanations). 

\begin{table}[h!]
\centering
\caption{\small{\textbf{Summary of definitions of over-parameterization in classical and quantum machine learning}. The labels `Q-optimization' and `Q-learning' denote when QNNs are applied to optimization and learning tasks, respectively. The labels `agnostic' and `informed' refer that the employed ansatz of QNN is agnostic of the explored problem and problem-dependent, respectively.  Notations $N_t$ and $n$ refer to the number of trainable parameters and the number of training examples, respectively. The notation $d_{\text{eff}}$ denotes the effective dimension of an $N$-qubit Hamiltonian. The notation $N$ represents the number of qubits.}}
\label{tab:over-para-sum}
\resizebox{0.85\textwidth}{!}{%
\begin{tabular}{|c|c|c|c|c|c|}
\hline
\begin{tabular}[c]{@{}c@{}}DNN \\  Refs.~\cite{zhang2021understanding,allen2019learning,allen2019convergence} \end{tabular} &
  \begin{tabular}[c]{@{}c@{}}Q-optimization\\ (agnostic)\\ Ref.~\cite{liu2023analytic} \end{tabular} &
  \begin{tabular}[c]{@{}c@{}}Q-optimization\\ (informed)\\ Refs.~\cite{larocca2021theory,you2022convergence,wang2022symmetric,sauvage2022building}\end{tabular} &
  \begin{tabular}[c]{@{}c@{}}Q-learning\\ (agnostic)\\ Ref.~\cite{you2023analyzing} \end{tabular} &
  \begin{tabular}[c]{@{}c@{}}Q-learning\\ (informed)\\ Ref.~\cite{schatzki2022theoretical} \end{tabular} &
  Our work \\ \hline
$N_t \gg n$ &
  $N_t\sim O(\exp(N))$ &
  $N_t \geq \Omega(d_{\text{eff}})$ &
  $N_t \sim O(\exp(N))$ &
  $N_t \sim O(poly(N))$ &
  \begin{tabular}[c]{@{}c@{}}$N_t >n$ \\ and\\  $N_t > O(\text{poly}(N))$\end{tabular} \\ \hline
\end{tabular}%
}
\end{table}

We will now elucidate how these diverse definitions of over-parameterization have inspired the definition of over-parameterization  in Definition 1 of the main text. Recall that our work considers both the trainability and the generalization of QCs on multi-class classification tasks, i.e., these two quantities are defined as $\mathsf{R}_{\text{ERM}}$ and $\mathsf{R}_{\text{Gene}}$ in Eq.~(2) of the main text. Hence, the notion of over-parameterization in our work should encompass both of these two aspects simultaneously.  
\begin{itemize}
	\item From the perspective of generalization, over-parameterization in our work refers to the condition where the number of training parameters $N_t$  is greater than the number of training examples $n$, i.e., $N_t > n$. This is in line with prior results in the context of deep learning theory \cite{zhang2021understanding,allen2019learning,allen2019convergence} and the generalization of quantum neural networks \cite{caro2021generalization,du2022efficient,gyurik2021structural,cai2022sample}, which have shown that the generalization error of QNNs is upper bounded by $O(\sqrt{N_t/n})$. Therefore, the over-parameterized regime of QNNs in terms of generalization is when $N_t > n$, as in this case, previous generalization bounds become vacuous and fail to explain the generalization ability of QNNs.  
	\item From the perspective of trainability, over-parameterization in our work refers to the scenario where  the adopted ansatz $U(\btheta)$ forms a 2-design, implying that the concentration of measures occurs \cite{popescu2006entanglement,ledoux2001concentration}. The condition of reaching over-parameterization depends on the layout of the selected ansatz. For instance, when $U(\btheta)$ corresponds to the hardware-efficient ansatz, which is supported by results related to barren plateaus \cite{mcclean2018barren, cerezo2020cost}, the over-parameterized regime of QNNs in terms of trainability can be defined as $N_t > O(\text{poly}(N))$, where $N$ represents the number of qubits.      
\end{itemize}
To this end, the over-parameterization in Definition 1 of the main text has a two-fold meaning: the over-parameterization regime of QCs is  $N_t>n$ in terms of generalization and when $U(\btheta)$ forms a 2-design in terms of trainability.

\section{Details related to the proof of Theorem 1}\label{append:thm:opt_learnability_QC-thm3}

For convenience, let us first recall the settings and notations introduced in the main text. When QCs are applied to accomplish the multi-class classification task, the training dataset $\mathcal{D}$ contains $n$ examples and the number of examples in each class is the same with $n=n_cK$. Moreover, the per-sample loss is specified as the mean square error. The dimension of feature states is larger than the number of classes, i.e., $2^D\geq K$.

\medskip
We next introduce the formal description of Theorem 1. Following notations of Eq.~(1) in the main text, we first consider two loss functions for QCs. The first loss function considers the tunable measure operators $\bm{o}$ and the  regularizer $\mathfrak{E}(h)$  whose explicit form is 
\begin{equation}\label{append:eqn:loss_mse_reg}
\Loss(\bm{\rho}, \bm{o}) =  \frac{1}{2n}\sum_{i=1}^{n_c}\sum_{k=1}^K\left([\Tr(\rhoik o^{(k)})]_{k=1:K} - \byik \right)^2 + \frac{\lambda_{\rho}}{2} \sum_{i=1}^{n_c}\sum_{k=1}^K \|\rhoik \|_F^2 +  \frac{\lambda_{o}}{2}\sum_{j=1}^K \|o^{(j)}\|_F^2,  
\end{equation}
where $\byik$ is the unit basis whose $k$-th entry is $1$ for $\forall i\in[n_c]$, $\forall k\in[K]$, and $\lambda_{\rho}$ and $\lambda_o$ refer to hyper-parameters of regularizer yielding $\lambda_o\leq n_c\lambda_{\rho}$ and $C_1:=K\sqrt{n_c\lambda_{o}\lambda_{\rho}}\leq 1$. The second loss function considers the fixed measure operator and the non-regularizer case whose explicit form is
 \begin{equation}\label{append:eqn:loss_mse_reg_no}
\Loss(\bm{\rho}) =  \frac{1}{2n}\sum_{i=1}^{n_c}\sum_{k=1}^K\left([\Tr(\rhoik o^{(k)})]_{k=1:K} - \byik \right)^2 + \frac{\lambda_{\rho}}{2} \sum_{i=1}^{n_c}\sum_{k=1}^K \|\rhoik \|_F^2.  
\end{equation}

Under the above settings, the formal statement of Theorem 1 is as follows. 
\begin{theorem-non}[Formal statement of Theorem 1] 
When QC is optimized under the loss function $\Loss(\bm{\rho}, \bm{o})$ in Eq.~(\ref{append:eqn:loss_mse_reg}), the global minimizer  $(\bm{\rho}^*, \bm{o}^*)=\min_{\bm{\rho}, \bm{o}}\Loss(\bm{\rho}, \bm{o})$  that can reach $\RERM=C_1^2/2$  satisfies the following properties:  
\begin{equation}\label{append:eqn:thm-1-condi-rho-reg}
(i) \bar{\rho}^{*(k)}:=\rho^{*(1,k)}=...= \rho^{*(n_c,k)}; ~(ii) \Tr(\bar{\rho}^{*(k)} \bar{\rho}^{*(k')})=  (1-C_1) \sqrt{\frac{\lambda_{o}}{n\lambda_{\rho}}} \delta_{k,k'};~(iii) o^{*(k)}=\sqrt{\frac{n\lambda_{\rho}}{\lambda_{o}}} \bar{\rho}^{*(k)}.
\end{equation}
Moreover, when the size of train dataset satisfies $n\gg  O(KN_{ge}\log \frac{KN_g}{\epsilon\delta})$, 
with probability $1-\delta$, the expected risk of this optimal QC tends to be $\ROPT(\hath_Q)=C_1^2/2$. 

When QC is optimized under the loss function $\Loss(\bm{\rho})$ in Eq.~(\ref{append:eqn:loss_mse_reg_no}) and the predefined $\{o^{(k)}\}$ are mutually orthogonal with each other, the global minimizer  $\bm{\rho}^*=\min_{\bm{\rho}, \bm{o}}\Loss(\bm{\rho}, \bm{o})$  that can reach $\RERM=0$  satisfies the following properties: 
\begin{equation}\label{append:eqn:thm-1-condi-rho}
(i) \bar{\rho}^{*(k)}:=\rho^{*(1,k)}=...= \rho^{*(n_c,k)}; ~(ii) \Tr(\bar{\rho}^{*(k)} \bar{\rho}^{*(k')})=  \frac{1}{B}\delta_{k,k'};~(iii) \Tr(\bar{\rho}^{*(k)} o^{(k')})=  \delta_{k,k'}.
\end{equation}
Moreover, when the size of train dataset satisfies $n\gg  O(KN_{ge}\log \frac{KN_g}{\epsilon\delta})$,  with probability $1-\delta$, the expected risk of this optimal QC tends to be zero, i.e., $\ROPT(\hath_Q)=0$.
\end{theorem-non}

\noindent\textbf{Remark}.  Due to the similar geometric properties for $(\bm{\rho}^*, \bm{o}^*)$ for the loss in Eq.~(\ref{append:eqn:loss_mse_reg}) and $\bm{\rho}^*$ for the loss in Eq.~(\ref{append:eqn:loss_mse_reg_no}), we mainly focus on the latter case throughout the whole study, which are efficiently implementable and are adopted by most QCs.

\subsection{Proof of Theorem 1}
Let us first illustrate the proof sketch and then provide the proof details.  Theorem 1 is achieved by separately analyzing $\RERM(\hath_Q)$ and $\RGENE(\hath_Q)$. For $\RERM(\hath_Q)$, we first consider the most general case in which both $\bm{\rho}$ and $\bm{o}$ are tunable as defined in Eq.~(\ref{append:eqn:loss_mse_reg}), with  $\hath_Q\equiv h_Q(\bm{\rho}^*, \bm{o}^*)$ and $(\bm{\rho}^*, \bm{o}^*)=\min_{\bm{\rho},\bm{o}}\Loss(\bm{\rho},\bm{o})$. The achieved results are summarized in the following lemma.    
\begin{lemma}[Informal]\label{prop:gene-Qc-uncons}
When the regularizer $\mathfrak{E}$ is considered and the global minimizer in Eq.~(\ref{append:eqn:loss_mse_reg}) has $\RERM(\hath_Q)= C_1^2/2$ with $C_1$ depending on the hyper-parameters in $\mathfrak{E}$, $(\bm{\rho}^*, \bm{o}^*)$ satisfies the three properties in Theorem 1.  
\end{lemma}
\noindent The achieved properties of $\bm{o}^*$ can be used as a priori to simplify QCs.  Particularly, according to the geometric properties of $\bm{o}$ derived in Lemma \ref{prop:gene-Qc-uncons}, we consider the fixed measurement operators $\bm{o}$ such that the set of measurements is complete and their vectorization spans the 2D-dimensional identity, following conventions of most QCs. The following lemma quantifies $\RERM(\hath_Q)$ when QC is optimized by the loss defined in Eq.~(\ref{append:eqn:loss_mse_reg_no}), where $\bm{o}$ is fixed under the above setting, $\mathfrak{E}=0$, and  $\hath_Q\equiv h_Q(\bm{\rho}^*, \bm{o})$ with $ \bm{\rho}^* =\min_{\bm{\rho}}\Loss(\bm{\rho})$.
\begin{lemma}[Informal]\label{thm:condit-NC-QNN}
When the predefined $\{o^{(k)}\}$ are mutually orthogonal with each other and  the global minimizer in Eq.~(\ref{append:eqn:loss_mse_reg}) has $\RERM(\hath_Q)=0$, $\bm{\rho}^* $ satisfies the three properties in Theorem 1. 
\end{lemma}
\noindent The proofs of Lemmas \ref{prop:gene-Qc-uncons} and \ref{thm:condit-NC-QNN} are given in SM~C and SM~D, respectively.

The rest part to prove Theorem 1 is analyzing the upper bound of $\RGENE(\hath_Q)$. Prior results cannot be used to prove Theorem 1,  since such bounds polynomially scale with the trainable parameters and become vacuous in the over-parameterized regime. To remedy this issue, we utilize the concept of algorithmic robustness \cite{Xu2010Robustness}.   
\begin{definition}[Robustness]\label{def:robustness}
A learning algorithm $\mathcal{A}$ is $(R, \nu(\cdot))$-robust with $R\in \mathbb{N}$ and $\nu(\cdot):\mathcal{Z}^n\rightarrow \mathbb{R}$, if $\mathcal{Z}$ can be partitioned into $R$ disjoint sets, denoted by $\{C_r\}_{r=1}^R$, such that the following holds for all $\mathcal{D}\subset \mathcal{Z}^n:\forall \bm{s} =(\bxi, y^{(i)}) \in \mathcal{D}$, $\forall \bm{z}=(\bm{x}, y)\in\mathcal{Z}$, $\forall r\in [R]$, \[\bm{s}, \bm{z} \in  \mathcal{C}_r \Rightarrow |l(h_{\mathcal{A}_{\mathcal{D}}}(\bxi), y^{(i)}) - l(h_{\mathcal{A}_{\mathcal{D}}}(\bm{x}), y)|\leq \nu(\mathcal{D}).\]
\end{definition}
\noindent Robustness measures how much the loss value can be varied with respect to the input space $\mathcal{Z}$. A higher robustness   admits  lower $R$, $\nu(\cdot)$, and $\RGENE$ \cite{Xu2010Robustness}.  The following lemma quantifies the upper bound of $\RGENE(\hath_Q)$ whose proof is given in SM~E.
\begin{lemma}\label{thm:gene_robust_QNN} Suppose the measure operator is bounded by $C_2$ with $\max_{k\in [K]}\|o^{(k)}\|\leq C_2$. Define $\epsilon$ as the tolerable error.  Following notations in Definition \ref{def:robustness}, the empirical QC is $(K (28N_{ge}/\epsilon)^{4^m N_{ge}}, 4L_1KC_2 \epsilon)$-robust, and with probability $1-\delta$ we have   
{\small\[\RGENE(\hath_Q)\leq 4L_1KC_2 \epsilon + 5\xi(\hath_Q)  \sqrt{\frac{|\mathcal{T}_{\mathcal{D}}|4^m N_{ge}\ln \frac{56KN_{ge}}{\epsilon\delta}}{n}},\]}
\noindent where $L_1$ is the Lipschitz constant of $\ell$ with respect to $h_Q$, $\mathcal{I}_r^{\mathcal{D}}=\{i\in[n]:\bm{z}^{(i)}\in \mathcal{C}_r\}$, $\xi(\hath):=\max_{\bm{z}\in \mathcal{Z}}(\ell(\hath, \bm{z}))$, and $\mathcal{T}_{\mathcal{D}}:=\{r\in [R]:|\mathcal{I}_r^{\mathcal{D}}|\geq 1\}$.
\end{lemma}
\noindent The achieved bound conveys two insights. First, it does not explicitly depend on the number of trainable parameters. This unlocks a new way to understand the generalization ability of QCs, especially for the over-parameterized ones. Next, it hints that a carefully designed $U_E$ can enhance performance of QCs \cite{caro2021encoding,hayashi2022efficient}. 

\noindent\textit{Remark}. The exact form of the first term in the generalization bound should be $4L_1KC_2 f(U(\bm{\theta})) \epsilon$ with $f(U(\bm{\theta})) \leq 1$ for any Ansatz. Therefore, for simplicity, we discard the term $f(U(\bm{\theta}))$. In addition, $\RGENE(\hath_Q)\rightarrow 0$  requires $n\gg |\mathcal{T}_{\mathcal{D}}|4^m N_{ge}$. Fortunately,  a reasonable value of $n$ is sufficient to warrant this condition, because in general  $m \leq 2$, $N_{ge}\propto |\bx|$, and $|\mathcal{T}_{\mathcal{D}}|$ is continuously decreased from $n$ to $K$ with respect to the reduced empirical loss

\medskip
Theorem 1 can be readily achieved by combining Lemmas \ref{prop:gene-Qc-uncons}, \ref{thm:condit-NC-QNN}, and \ref{thm:gene_robust_QNN}.
\begin{proof}[Proof of Theorem 1]
  Following the definition of expected risk (Eq.~(2) in the main text) and the results in Lemma  3, with probability $1-\delta$, the expected risk of an optimal empirical QC is upper bounded by
  \begin{equation}
    \ROPT(\hath_Q)   \leq \RERM(\hath_Q) + 4L_1KC_2 \epsilon + 3\xi(\hath)\sqrt{\frac{|\mathcal{T}_{\mathcal{D}}|4^m N_{ge}\ln(56KN_{ge}/(\epsilon\delta))}{n}} + \xi(\hath)\frac{2|\mathcal{T}_{\mathcal{D}}|4^m N_{ge}\ln(56KN_{ge}/(\epsilon\delta))}{n}.
  \end{equation}
Then, for the loss function defined in Eq.~(\ref{append:eqn:loss_mse_reg}), when $(\bm{\rho}^*, \bm{o}^*)$ satisfies Eq.~(\ref{append:eqn:thm-1-condi-rho-reg}), Lemma \ref{prop:gene-Qc-uncons} warrants $\RERM(\hath_Q)=C_1^2/2$, which gives
 \begin{equation}
    \ROPT(\hath_Q)   \leq  \frac{C_1^2}{2} + 4L_1KC_2 \epsilon + 3\xi(\hath)\sqrt{\frac{|\mathcal{T}_{\mathcal{D}}|4^m N_{ge}\ln(56KN_{ge}/(\epsilon\delta))}{n}} + \xi(\hath)\frac{2|\mathcal{T}_{\mathcal{D}}|4^m N_{ge}\ln(56KN_{ge}/(\epsilon\delta))}{n}.
  \end{equation}

Similarly, for the loss function defined in Eq.~(\ref{append:eqn:loss_mse_reg_no}), when $\bm{\rho}^*$ satisfies Eq.~(\ref{append:eqn:thm-1-condi-rho}), Lemma \ref{thm:condit-NC-QNN} warrants $\RERM(\hath_Q)=0$, which gives
 \begin{equation}
    \ROPT(\hath_Q)   \leq   4L_1KC_2 \epsilon + 3\xi(\hath)\sqrt{\frac{|\mathcal{T}_{\mathcal{D}}|4^m N_{ge}\ln(56KN_{ge}/(\epsilon\delta))}{n}} + \xi(\hath)\frac{2|\mathcal{T}_{\mathcal{D}}|4^m N_{ge}\ln(56KN_{ge}/(\epsilon\delta))}{n}.
  \end{equation}
       
This bound can be further simplified when the training of QC is perfect. Note that  Condition (i) implies  $|\mathcal{T}_{\mathcal{D}}|=K$, since all feature states from the same class collapse to the same point. Meanwhile, since $\xi(\hath)$ and $C_2$ are bounded, and $m$ and $\epsilon$ are small constant, we can conclude that when $n \gg O(KN_{ge}\log(KN_g/(\epsilon\delta)))$, the expected risk can approach to zero. 
\end{proof}

\section{Proof of Lemma \ref{prop:gene-Qc-uncons}}\label{append:sec:pro1}

In this section, we derive the geometric properties of the global optimizer under the unconstraint loss function $\Loss(\bm{\rho}, \bm{o})$ defined in Eq.~(\ref{append:eqn:loss_mse_reg}), where both $\bm{\rho}$ and $\bm{o}$ are tunable and the regularization term is considered.  Denote the global optima as $(\bm{\rho}^*, \bm{o}^*)=\min_{\bm{\rho},\bm{o}}\Loss(\bm{\rho},\bm{o})$ and the empirical QC as $\hath_Q\equiv h_Q(\bm{\rho}^*, \bm{o}^*)$. The restatement of Lemma \ref{prop:gene-Qc-uncons} is as follows.
\begin{lemma-non}[Formal statement of Lemma \ref{prop:gene-Qc-uncons}]
Following notations and settings in Eq.~(\ref{append:eqn:loss_mse_reg}),  the global minimizer ($\bm{\rho}^*, \bm{o}^*$) of $\Loss(\bm{\rho}, \bm{o})$ satisfies for $\forall k,k'\in[K]$: 
\begin{eqnarray}
&&(i) \bar{\rho}^{*(k)}:=\rho^{*(1,k)}=\cdots= \rho^{*(n_c,k)};\nonumber \\ 
&&(ii) \Tr(\bar{\rho}^{*(k)} \bar{\rho}^{*(k')})=(1-C_1) \sqrt{\frac{\lambda_{o}}{n\lambda_{\rho}}} \delta_{k,k'}; \nonumber\\ 
&&(iii) o^{*(k)}=\sqrt{\frac{n\lambda_{\rho}}{\lambda_{o}}} \bar{\rho}^{*(k)}.
\end{eqnarray} 
The corresponding empirical risk is  $\RERM(\hath_Q)=C_1^2/2$.
\end{lemma-non}
\begin{proof}[Proof of Lemma \ref{prop:gene-Qc-uncons}]
 Conceptually, the global optimizer can be identified by lower bounding $\Loss(\bm{\rho}, \bm{o})$, where the equality conditions of $\bm{\rho}$ amount to the properties of global minimizer. In particular, the lower bound of  $\Loss(\bm{\rho}, \bm{o})$ yields  
 \begin{eqnarray}\label{append:eqn:prop1-1}
 && \frac{1}{2Kn_c}\sum_{i=1}^{n_c}\sum_{k=1}^K\left([\Tr(\rhoik o^{(j)})]_{j=1:K} - \byik \right)^2 + \frac{\lambda_{\rho}}{2} \sum_{i=1}^{n_c}\sum_{k=1}^K \|\rhoik \|_F^2 +  \frac{\lambda_{o}}{2}\sum_{j=1}^K \|o^{(j)}\|_F^2 \nonumber\\
 \geq && \frac{1}{2Kn_c}\sum_{i=1}^{n_c}\sum_{k=1}^K\left(\Tr(\rhoik o^{(k)}) - 1 \right)^2 + \frac{\lambda_{\rho}}{2} \sum_{i=1}^{n_c}\sum_{k=1}^K \|\rhoik \|_F^2 +  \frac{\lambda_{o}}{2}\sum_{j=1}^K \|o^{(j)}\|_F^2 \nonumber\\
 = && \frac{1}{2Kn_c}\sum_{k=1}^K \sum_{i=1}^{n_c} n_c \frac{1}{n_c} \left(\Tr(\rhoik o^{(k)}) - 1 \right)^2 + \frac{\lambda_{\rho}}{2} \sum_{k=1}^K \sum_{i=1}^{n_c} n_c \frac{1}{n_c} \|\rhoik \|_F^2 +  \frac{\lambda_{o}}{2}\sum_{j=1}^K \|o^{(j)}\|_F^2  \nonumber\\
 \geq && \frac{1}{2K}\sum_{k=1}^K  \left(\Tr\left(\sum_{i=1}^{n_c}  \frac{1}{n_c}  \rhoik o^{(k)} \right) - 1 \right)^2 + \frac{\lambda_{\rho}}{2} \sum_{k=1}^K  n_c  \left\| \sum_{i=1}^{n_c} \frac{1}{n_c} \rhoik \right\|_F^2 +  \frac{\lambda_{o}}{2}\sum_{j=1}^K \|o^{(j)}\|_F^2,
 \end{eqnarray}
  where the first inequality uses the fact  $\|\bm{a}-\bm{b}\|^2 = \sum_i (\bm{a}^{(i)} - \bm{b}^{(i)})^2 \geq (\bm{a}^{(k)} - \bm{b}^{(k)})^2$ and the $k$-th entry of $\byik$ equals to $1$, and the second inequality comes from the Jensen's inequality $f(\mathbb{E}(x))\leq \mathbb{E}(f(x))$. The equality condition of the first inequality holds if and only if 
  \[\Tr\left(\rhoik o^{(j)}\right)=0,~\left(\forall j\in [K]\setminus \{k\}\right) \land \left(\forall i\in [n_c]\right);\] 
and the equality condition of the second inequality holds if and only if 
  \[ \rho^{(1,k)}=\cdots=\rhoik=\cdots =\rho^{(n_c,k)},~\forall k\in[K]. \] 
Denote the mean of the feature state for the $k$-th class as $\bar{\rho}^{(k)}= \sum_{i=1}^{n_c} \frac{1}{n_c} \rhoik$ for $\forall k\in[K]$. The above two equality conditions suggest that the global minimizer $(\bm{\rho}^*, \bm{o}^*)$ satisfies  
\begin{eqnarray}\label{append:eqn:prop1:condi-1}  
&& \bar{\rho}^{*(k)} \equiv \rho^{*(1,k)}   = \cdots =\rho^{*(n_c,k)},~ \forall k\in[K] \nonumber\\
  && \Tr(\bar{\rho}^{*(k)} o^{*(j)})=0,~ \forall j\in [K]\setminus \{k\}.  
\end{eqnarray}
To this end, we obtain Conditions (i) in Lemma \ref{prop:gene-Qc-uncons}, which describe the geometric properties of  $\bm{\rho}^*$, i.e., 
\begin{eqnarray}
(i) \bar{\rho}^{*(k)}:=\rho^{*(1,k)}=\cdots= \rho^{*(n_c,k)}.
\end{eqnarray}

The next part of the proof is showing that the global minimizer satisfies Condition (iii). Combining  Eqs.~(\ref{append:eqn:prop1-1}) and (\ref{append:eqn:prop1:condi-1}), the lower bound of the loss function in Eq.~(\ref{append:eqn:prop1-1}) follows  
\begin{eqnarray}
  &&  \Loss(\bm{\rho}, \bm{o}) \nonumber\\
  \geq && \frac{1}{2K }\sum_{k=1}^K    \left(\Tr\left(\bar{\rho}^{(k)} o^{(k)} \right) - 1 \right)^2 + \frac{\lambda_{\rho}}{2} \sum_{k=1}^K  n_c  \left\| \bar{\rho}^{(k)}  \right\|_F^2 +  \frac{\lambda_{o}}{2}\sum_{j=1}^K \|o^{(j)}\|_F^2 \nonumber\\
  = && \frac{1}{2K }\sum_{k=1}^K    \left(\Tr\left(\bar{\rho}^{(k)} o^{(k)} \right) - 1 \right)^2 + \frac{\lambda_{\rho}}{2} K  \sum_{k=1}^K \frac{1}{K} n_c  \left\| \bar{\rho}^{(k)}  \right\|_F^2 +  \frac{\lambda_{o}}{2} K  \sum_{j=1}^K \frac{1}{K} \|o^{(j)}\|_F^2  \nonumber\\
  \geq && \frac{1}{2}    \left(\sum_{k=1}^K \frac{1}{K} \Tr\left(\bar{\rho}^{(k)} o^{(k)} \right) - 1 \right)^2 + \frac{\lambda_{\rho}}{2} K   n_c  \left\|\sum_{k=1}^K\frac{1}{K} \bar{\rho}^{(k)}  \right\|_F^2 +  \frac{\lambda_{o}}{2}K \|\sum_{j=1}^K \frac{1}{K} o^{(j)}\|_F^2,
\end{eqnarray}
where the second inequality comes from the Jensen's inequality and the equality condition holds if and only if for $\forall k,k'\in[K]$, 
\begin{equation}\label{append:eqn:prop1:condi-2}
\Tr\left(\bar{\rho}^{(k)} o^{(k)} \right) = \Tr\left(\bar{\rho}^{(k')} o^{(k')} \right),~\|\bar{\rho}^{(k)}\|_F= \|\bar{\rho}^{(k')}\|_F,~ \|o^{(k)}\|_F= \|o^{(k')}\|_F. 
\end{equation}
Then, supported by the inequality  $a+b\geq 2\sqrt{ab}$, the loss $\Loss(\bm{\rho}, \bm{o})$ can be further lower bounded by
\begin{eqnarray}\label{append:eqn:prop1-2}
  && \frac{1}{2}    \left( \Tr\left(\bar{\rho}^{(k)} o^{(k)} \right) - 1 \right)^2 + \frac{\lambda_{\rho}}{2} K   n_c  \left\|  \bar{\rho}^{(k)}  \right\|_F^2 +  \frac{\lambda_{o}}{2}K \|  o^{(j)}\|_F^2  \nonumber\\
  \geq && \frac{1}{2}    \left( \Tr\left(\bar{\rho}^{(k)} o^{(k)} \right) - 1 \right)^2 + K\sqrt{n_c \lambda_{o} \lambda_{\rho}}  \left\| \bar{\rho}^{(k)}  \right\|_F \|o^{(j)}\|_F, 
\end{eqnarray}
where the equality condition holds if and only if 
\begin{equation}\label{append:eqn:prop1:condi-3}
  \lambda_{o}    \|o^{(j)}\|_F^2 =  n_c  \lambda_{\rho}   \left\| \bar{\rho}^{(k)}  \right\|_F^2, \forall k \in [K].
\end{equation}
Note that the requirements $C_1\leq 1$ and $\lambda_o\leq n_c\lambda_{\rho}$ in Lemma \ref{prop:gene-Qc-uncons} imply $\|\bar{\rho}^{*(k)}\|\leq 1$ and hence ensure that $\bar{\rho}^{*(k)}$ is a meaningful quantum state for $\forall k\in[K]$.

 Since $\Tr\left(\bar{\rho}^{(k)} o^{(k)} \right)=\|\bar{\rho}^{(k)}\| \|o^{(k)}\|\cos(\angle(\rho^{(k)}, o^{(k)})) $, the lower bound of $\Loss(\bm{\rho}, \bm{o})$ in Eq.~(\ref{append:eqn:prop1-2}) is equivalent to 
\[\frac{1}{2}    \left( \|\bar{\rho}^{(k)}\| \|o^{(k)}\|\cos(\angle(\rho^{(k)}, o^{(k)}))  - 1 \right)^2 + C_1  \left\| \bar{\rho}^{(k)}  \right\|_F \|o^{(j)}\|_F.\] 
Define $\|\bar{\rho}^{(k)}\| \|o^{(k)}\|=a$ and $\angle(\rho^{(k)}, o^{(k)})=\alpha$. The above equation is described by the function $f(a,\alpha)=(a\cos\alpha -1)^2/2+C_1a$ and its minimum is $C_1-C_1^2/2$ when $\alpha^*=0$ and $a^*=1-C_1$. The derivation is as follows. Since $a>0$ and its maxima is unbounded, we first consider the case $0<a < 1$. In this case, the minimum of $f(a,\alpha)$ is $C_1-C_1^2/2$ with $\alpha^*=0$ and $a^*=1-C_1$. Otherwise, when $a\geq 1$, the minimum of $f(a,\alpha)$ is $C_1$ with $\alpha^*=\arccos(1/a)$ and $a^*=1$. Note that the minimum value of $f(a,\alpha)$ in the second case is always larger than that of the first case. Therefore, the minimum of $f(a,\alpha)$ is $C_1-C_1^2/2$ with $\alpha^*=0$ and $a^*=1-C_1$. Combining the observation that $\bar{\rho}^{*(k)}$ and $o^{(k)}$ are in the same direction with Eq.~(\ref{append:eqn:prop1:condi-3}), we achieve Condition (iii), i.e.,
\[o^{*(k)}=\sqrt{\frac{n_c\lambda_{\rho}}{\lambda_{o}}} \bar{\rho}^{*(k)}.\]

The last part is proving Condition (ii).  Combining the result $\|\bar{\rho}^{*(k)}\| \|o^{*(k)}\|=1-C_1$ for $\forall k\in[K]$ with Eq.~(\ref{append:eqn:prop1:condi-1})  and Condition (iii), we immediately  obtain condition (ii), i.e.,
\begin{equation}
	(ii) \sqrt{\frac{n_c\lambda_{\rho}}{\lambda_{o}}}\|\rho^{*(k)}\|\|\rho^{*(k')}\|=(1-C_1)\delta_{k,k'}   \Rightarrow  \Tr(\bar{\rho}^{*(k)} \bar{\rho}^{*(k')}) = (1-C_1) \sqrt{\frac{\lambda_{o}}{n_c\lambda_{\rho}}} \delta_{k,k'}.
\end{equation} 

To summarize, given the global optima satisfying the above three conditions, the corresponding empirical risk is
\begin{equation}
\RERM(\hath_Q) = \frac{1}{2n}\sum_{i=1}^{n_c}\sum_{k=1}^K\left([\Tr(\rho^{*(i,k)}o^{*(k)})]_{k=1:K} - \byik \right)^2   = \frac{C_1^2}{2}
\end{equation}
\end{proof}

\section{Results related to Lemma \ref{thm:condit-NC-QNN}} \label{append:sec:proof-thm1} 

This section is composed of two parts. In SM~\ref{append:sec1:subsec:proof-thm-1}, we present the proof of Lemma \ref{thm:condit-NC-QNN}. In SM~\ref{append:sec1:subsec-assump-thm1}, we explain that the requirements in Lemma \ref{thm:condit-NC-QNN} are mild.

\subsection{Proof of Lemma \ref{thm:condit-NC-QNN}}\label{append:sec1:subsec:proof-thm-1}
Different from Lemma \ref{prop:gene-Qc-uncons}, here we focus the setting such that the regularization term is set as $\mathfrak{E}=0$ and the operator $\bm{o}$ is predefined. The explicit form of the loss function $\Loss$ is defined in Eq.~(\ref{append:eqn:loss_mse_reg_no}). Denote the optimal feature states $\bm{\rho}^*=\min_{\bm{\rho}}\Loss(\bm{\rho})$, we quantify the value of $\RERM(\hath_Q)$ with $\hath_Q\equiv h_Q(\bm{\rho}^*)$. 

 We emphasize that the modifications of $\mathfrak{E}$ and $\bm{o}$ allow a lower optimal empirical risk. Recall the results of Lemma \ref{prop:gene-Qc-uncons}. In the most general case, the optimal empirical risk depends on the regularization term, i.e., $\RERM(\hath_Q)\rightarrow C_1^2/2$. The dependance on $C_1$ motivates us to explore the empirical risk of QC when $\mathfrak{E}=0$. Furthermore, Condition (iii) in Lemma \ref{prop:gene-Qc-uncons} delivers the crucial properties of the optimal measure operator, i.e., the optimal measure operators are orthogonal with each other. Such properties contribute to construct a more effective QCs. Instead of optimizing, the measure operator $\bm{o}$ can be predefined by inheriting the properties proved in Lemma \ref{prop:gene-Qc-uncons}, that is, $\bm{o}$ are required to span the space $\mathbb{C}^{2^D\times 2^D}$ and satisfy $\Tr(o^{(k)}o^{(k')})=B\delta_{k,k'}$ with $B\geq 1$ being a constant. Notably, these requirement are mild, covering frequently used measures such as computational basis and Pauli-based measures, as explained in SM~\ref{append:sec1:subsec-assump-thm1}.

\begin{lemma-non}[Formal statement of Lemma \ref{thm:condit-NC-QNN}]\label{append:thm:condit-NC-QNN}
  Suppose that the adopted measure operator $\bm{o}$ spans the space $\mathbb{C}^{2^D\times 2^D}$ and satisfies $\Tr(o^{(k)}o^{(k')})=B\delta_{k,k'}$ where $B\geq 1$ is a constant. The empirical risk of $\hath_Q$ is $\RERM(\hath_Q)=0$ when the global minimizer $\bm{\rho}^*$ satisfies
\begin{equation}
(i) \bar{\rho}^{*(k)}:=\rho^{*(1,k)}=...= \rho^{*(n_c,k)}; ~(ii) \Tr(\bar{\rho}^{*(k)} \bar{\rho}^{*(k')})=  \frac{1}{B}\delta_{k,k'};~(iii) \Tr(\bar{\rho}^{*(k)} o^{(k')})=  \delta_{k,k'}.
\end{equation}

\end{lemma-non}

\begin{proof}[Proof of Lemma \ref{thm:condit-NC-QNN}] 
 
The concept of the proof is analogous to Lemma \ref{prop:gene-Qc-uncons}, i.e., the global optimizer is identified by lower bounding the loss $\Loss(\bm{\rho})$.  To this end, the lower bound of $\Loss(\bm{\rho})$ yields  
 \begin{eqnarray}\label{append:thm1:eqn_obj}
 && \frac{1}{2Kn_c}\sum_{i=1}^{n_c}\sum_{k=1}^K\left([\Tr(\rhoik o^{(j)})]_{j=1:K} - \byik \right)^2  \nonumber\\
 \geq && \frac{1}{2Kn_c}\sum_{i=1}^{n_c}\sum_{k=1}^K\left(\Tr(\rhoik o^{(k)}) - 1 \right)^2  \nonumber\\
 = && \frac{1}{2Kn_c}\sum_{k=1}^K \sum_{i=1}^{n_c} n_c \frac{1}{n_c} \left(\Tr(\rhoik o^{(k)}) - 1 \right)^2   \nonumber\\
 \geq && \frac{1}{2K}\sum_{k=1}^K  \left(\Tr\left(\sum_{i=1}^{n_c}  \frac{1}{n_c}  \rhoik o^{(k)} \right) - 1 \right)^2,
 \end{eqnarray}
  where the first inequality uses the facts $n=Kn_c$, $\|\bm{a}-\bm{b}\|^2 = \sum_i (\bm{a}^{(i)} - \bm{b}^{(i)})^2 \geq (\bm{a}^{(k)} - \bm{b}^{(k)})^2$, and only the $k$-th entry of $\byik$ equals to $1$, and the second inequality comes from the Jensen's inequality $\mathbb{E}(f(x))\geq f(\mathbb{E}(x))$ when $f(\cdot)$ is convex. Note that the equality condition of the first inequality holds if and only if 
  \[\Tr(\rhoik o^{(j)})=0, \left(\forall j\in [K]\setminus \{k\}\right) \land \left(\forall i\in [n_c]\right);\] 
  And the equality condition of the second inequality holds if and only if 
  \[ \rho^{(1,k)}=\cdots=\rhoik=\cdots =\rho^{(n_c,k)}, \forall k\in[K]. \] 
Denote the mean of the feature state for the $k$-th class as $\bar{\rho}^{(k)}= \sum_{i=1}^{n_c} \frac{1}{n_c} \rhoik$ for $\forall k\in[K]$. The above two equality conditions suggest that the global minimizer yields 
\begin{eqnarray}\label{append:eqn:thm1:condit-global1}
  && \bar{\rho}^{*(k)} \equiv \rho^{*(1,k)}   = \cdots =\rho^{*(n_c,k)},~\forall k\in[K] \\
  && \Tr(\bar{\rho}^{*(k)} o^{(j)})=0, \forall j\in [K]\setminus \{k\}. \label{append:thm1:eqn_obj_cond3}
\end{eqnarray}

Combining Eqs.~(\ref{append:thm1:eqn_obj})-(\ref{append:thm1:eqn_obj_cond3}), the lower bound of the loss function $\Loss(\bm{\rho})$ satisfies 
\begin{eqnarray}\label{append:thm1:eqn_obj_cond2}
  && \frac{1}{2K }\sum_{k=1}^K    \left(\Tr\left(\bar{\rho}^{(k)} o^{(k)} \right) - 1 \right)^2    
  \geq  \frac{1}{2}    \left(\sum_{k=1}^K \frac{1}{K} \Tr\left(\bar{\rho}^{(k)} o^{(k)} \right) - 1 \right)^2,
\end{eqnarray}
where the inequality comes from the Jensen's inequality and the equality condition holds if and only if $\forall k,k'\in[K]$, 
\begin{equation}\label{append:eqn:thm1:condit-global2}
  \Tr\left(\bar{\rho}^{(k)} o^{(k)} \right) = \Tr\left(\bar{\rho}^{(k')} o^{(k')} \right). 
\end{equation}
Supported by Eq.~(\ref{append:eqn:thm1:condit-global2}), we can further lower bound $\Loss(\bm{\rho})$ with
\begin{equation}
  \frac{1}{2}    \left(\Tr\left(\bar{\rho}^{(k)} o^{(k)} \right) - 1 \right)^2 \geq 0,
\end{equation}
where the equality condition is achieved when  $\Tr(\bar{\rho}^{(k)} o^{(k)} )=1$ for $\forall k\in[K]$. 

Taken together, the global optimizer $\bm{\rho}^*$ should satisfy Condition (i)\&(iii) in Lemma \ref{thm:condit-NC-QNN}, where
\begin{eqnarray}
  &&(i)\bar{\rho}^{*(k)}:=\rho^{*(1,k)}=...= \rho^{*(n_c,k)}; \nonumber\\
  &&(iii) \Tr(\bar{\rho}^{*(k)} o^{(k')})=  \delta_{k,k'}. 
\end{eqnarray}

We last prove that Condition (iii) and the requirements of $\bm{o}$ lead to Condition (ii). In particular, denote the vectorization of $\rho^{*(k)}$ and $o^{(k)}$ as $|\rho^{*(k)}\rangle\rangle$ and $|o^{(k)}\rangle\rangle$, respectively. Condition (iii) can be rewritten as
\begin{equation}
  \Big \langle \Big\langle \bar{\rho}^{*(k)},  o^{(k')}\Big \rangle \Big\rangle=\delta_{k,k'}.
\end{equation}
Moreover, since  the set of measure operators $\{o^{(k)}\}$ is required to be complete in the space of $\mathbb{C}^{2^D}$ and  $\Tr(o^{(k)}o^{(k')})=B\delta_{k,k'}$ with $B\geq 1$ for $\forall k,k'\in [K]$, we have  
\[\sum_{k} \Big|o^{(k)}\Big \rangle \Big\rangle \Big \langle \Big\langle o^{(k)} \Big|  =B \mathbb{I}_{2^D}. \] 
Then, Condition (ii) can be derived as follows, i.e.,
\begin{eqnarray}
  && \Tr(\rho^{*(k)}\rho^{*(k')}) \nonumber\\
  = && \langle\langle \bar{\rho}^{*(k)} | \mathbb{I}_{2^D} |\rho^{*(k')}\rangle\rangle \nonumber \\
  = && \frac{1}{B} \left\langle \left\langle \bar{\rho}^{*(k)} \Big| \sum_{k''} |o^{(k'')}\rangle\rangle \langle\langle o^{(k'')}| \Big|\rho^{*(k')} \right\rangle \right\rangle \nonumber\\
  = && \frac{1}{B} \left\langle \left\langle \bar{\rho}^{*(k)} \Big| |o^{(k)}\rangle\rangle \langle\langle o^{(k)}| \Big|\rho^{*(k')} \right\rangle \right\rangle + \left\langle \left\langle \bar{\rho}^{*(k)} \Big| \sum_{k''\neq k} \ket{o^{(k'')}}\bra{o^{(k'')}} \Big|\rho^{*(k')} \right\rangle \right\rangle \nonumber\\
  = && \frac{1}{B}\delta_{k,k'}.
\end{eqnarray}
\end{proof}

\subsection{Requirement of $\bm{o}$ used in Lemma 2}\label{append:sec1:subsec-assump-thm1}
Here we elucidate that the requirements adopted in Lemma 2, i.e., $\bm{o}$ spans the complex space $2^D\times 2^D$ and satisfies $\Tr(o^{(k)}o^{(k')})=B\delta_{k,k'}$ with $B\geq 1$, are mild. Specifically, the employed measurements in most QNN-based classifiers satisfy these requirements, including the computational basis measurements and Pauli measurements. When these measure operators are applied, the feature states of the optimal QCs form the general simplex ETF in Definition \ref{def:NC}.

\noindent  \textit{\underline{Computational basis measurements}.} In this setting, the local measurement $o^{(k)}$ is set as $\ket{k}\bra{k}$ with $\ket{k}$ being the $k$-th computational basis for $\forall k\in[K]$. When $2^D=K$, $\{\ket{k}\}$ spans the whole space of $\mathbb{C}^{2^D\times 2^D}$ and we have $\Tr(o^{(k)}o^{(k')})=|\braket{k|k'}|^2=\delta_{k,k'}$ with $B=1$. The assumptions are satisfied. 

\noindent  \textit{\underline{Pauli measurements}.} Denote the Pauli operation applied to the $i$-th qubit as $P^{(i)}_a$ with $a\in\{X,Y,Z,I\}$ for $\forall i\in[D]$. Then, there are in total $4^D$ Pauli strings $P=\otimes_{i=1}^D P^{(i)}_a$ that form a orthogonal basis for the space $\mathbb{C}^{2^D\times 2^D}$. With setting $2^D=K$, each $o^{(k)}$ corresponds to one Pauli string with  $\Tr(o^{(k)}o^{(k')})=K\delta_{k,k'}$ with $B=K$.  

\section{Proof of Lemma \ref{thm:gene_robust_QNN}}\label{append:sec:gene_QNN}
For elucidating, let us  restate Lemma \ref{thm:gene_robust_QNN} below and introduce the proof sketch before moving on to present the proof details. 
\begin{lemma-non}[Formal statement of Lemma \ref{thm:gene_robust_QNN}]
Given a QC defined in Eq.~(3), let $\mathcal{E}$ be a quantum channel with 
\begin{equation}\label{append:eqn:QC_thm2}
  h_Q(\bx, U(\btheta), O^{(k)})\equiv \Tr(o^{(k)}\mathcal{E}(\sigma(\bx))),~\forall k\in[K].
\end{equation}
Suppose the measure operator follows $\max_{k\in [K]}\|o^{(k)}\|\leq C_2$.  The explicit form of the encoding unitary follows $U_E(\bx)=\prod_{g=1}^{N_g}u_g(\bx) \in\mathcal{U}(2^N)$ with the $g$-th  quantum gate $u_g(\bx) \in\mathcal{U}(2^m)$ operating with at most $m$ qubits with $m\leq N$ and $N_g$ gates consisting of $N_{ge}$ variational gates and  $N_g-N_{ge}$ fixed gates.

Following above notations and  Definition \ref{def:robustness}, the empirical QC is $(K (\frac{28N_{ge}}{\epsilon})^{4^m N_{ge}}, 4L_1KC_2 \epsilon)$-robust and with probability $1-\delta$, its generalization error  yields   
\[\RGENE(\hath)\leq 4L_1KC_2 \epsilon +  3\xi(\hath)\sqrt{\frac{|\mathcal{T}_{\mathcal{D}}|4^m N_{ge}\ln(56KN_{ge}/(\epsilon\delta))}{n}} + \xi(\hath)\frac{2|\mathcal{T}_{\mathcal{D}}|4^m N_{ge}\ln(56KN_{ge}/(\epsilon\delta))}{n},\] 
where $L_1$ is the Lipschitz constant of the per-sample loss $\ell$ with respect to $h$, $\mathcal{I}_r^{\mathcal{D}}=\{i\in[n]:\bm{z}^{(i)}\in \mathcal{C}_r\}$, $\xi(\hath):=\max_{\bm{z}\in \mathcal{Z}} \ell(\hath, \bm{z})$, and $\mathcal{T}_{\mathcal{D}}:=\{r\in [R]:|\mathcal{I}_r^{\mathcal{D}}|\geq 1\}$.
\end{lemma-non}
The proof of Lemma 3 is established on the following lemma, which leverages the algorithmic robustness to quantify the upper bound of the generalization error.
\begin{lemma}[Theorem 1, \cite{kawaguchi2022robustness}]\label{append:lem:gene-robust}
If the learning algorithm $\mathcal{A}$ is $(R,  \nu(\cdot))$-robust with $\{\mathcal{C}_r\}_{r=1}^R$, then for any $\delta > 0$, with probability at least $1-\delta$ over an i.i.d drawn of $n$ samples $\mathcal{D} = \{\bm{z}^{(i)}\}_{i=1}^n$ with $\bm{z}^{(i)}=(\bxi, y^{(i)})$, the returned hypothesis $\hath$ by $\mathcal{A}$ on $\mathcal{D}$ satisfies
\begin{equation}
  \RGENE(\hath)\leq \nu(\mathcal{D}) + \xi(\hath)\left((\sqrt{2}+1)\sqrt{\frac{|\mathcal{T}_{\mathcal{D}}|\ln(2R/\delta)}{n}} + \frac{2|\mathcal{T}_{\mathcal{D}}|\ln(2R/\delta)}{n}\right),
\end{equation}
where $\mathcal{I}_r^{\mathcal{D}}=\{i\in[n]:\bm{z}^{(i)}\in \mathcal{C}_r\}$, $\xi(\hath):=\max_{\bm{z}\in \mathcal{Z}}(\ell(\hath, \bm{z}))$, and $\mathcal{T}_{\mathcal{D}}:=\{r\in [R]:|\mathcal{I}_r^{\mathcal{D}}|\geq 1\}$. 
\end{lemma}
\noindent The above result hints that given a hypothesis $\hath$, its generalization error is upper bounded by the disjoint sets $\{\mathcal{C}_r\}_{r=1}^R$, where a lower cardinality $R$ allows a lower generalization error. A natural approach to realize these disjoint partitions is covering number \cite{Xu2010Robustness}.   
\begin{definition}[Covering number, \cite{mohri2018foundations}]\label{def:cov-num}
Given a metric space $(\mathcal{U}, \|\cdot\|)$, the covering number $\mathcal{N}(\mathcal{U}, \epsilon, \|\cdot \|)$ denotes the least cardinality of any subset $\mathcal{V} \subset \mathcal{U}$ that covers $\mathcal{U}$ at scale $\epsilon$ with a norm $ \|\cdot \|$, i.e., $\sup_{A\in \mathcal{U}} \min_{B\in \mathcal{V}} \|A - B \|\leq \epsilon$.  
\end{definition}
\noindent In conjunction with Lemma \ref{append:lem:gene-robust} and Definition \ref{def:cov-num}, the analysis of $\RGENE(\hath)$ of an $N$-qubit QC amounts to quantifying the covering number of the space of the input quantum states, i.e.,
\begin{equation}\label{append:eqn:input_state_space}
  \mathcal{X}_Q =\left\{U_E(\bm{x})(\ket{0}\bra{0})^{\otimes N}U_E(\bm{x})^{\dagger} \big| \bx\in \mathcal{X} \right\}.
\end{equation} 
The following lemma connects the robustness and covering number of $\mathcal{X}_Q$ of QCs whose proof is provided in Sec.~\ref{append:subsec:lemma-rob-cov}.
\begin{lemma}\label{lem:lemma-rob-cov}
Following the settings in Eqs.~(\ref{append:eqn:QC_thm2})-(\ref{append:eqn:input_state_space}), the corresponding  QC  is $(K (\frac{28N_{ge}}{\epsilon})^{4^m N_{ge}}, 4L_1KC_2\|\mathcal{E}\|_{\diamond}\epsilon )$-robust.
\end{lemma}
We are now ready to prove Lemma 3. 
\begin{proof}[Proof of Lemma 3]
The generalization error bound can be acquired by combining Lemmas \ref{append:lem:gene-robust} and \ref{lem:lemma-rob-cov}, i.e.,
\begin{eqnarray}
\RGENE(\hath)\leq &&  4L_1KC_2\|\mathcal{E}\|_{\diamond}\epsilon + \xi(\hath)\left((\sqrt{2}+1)\sqrt{\frac{|\mathcal{T}_{\mathcal{D}}| \ln(2 K (\frac{28N_{ge}}{\epsilon})^{4^m N_{ge}}/\delta)}{n}} + \frac{2|\mathcal{T}_{\mathcal{D}}| \ln(2 K (\frac{28N_{ge}}{\epsilon})^{4^m N_{ge}}/\delta)}{n}\right)\nonumber\\
\leq &&  4L_1KC_2\|\mathcal{E}\|_{\diamond}\epsilon + \xi(\hath)\left(3\sqrt{\frac{|\mathcal{T}_{\mathcal{D}}|4^m N_{ge}\ln(56KN_{ge}/(\epsilon\delta))}{n}} + \frac{2|\mathcal{T}_{\mathcal{D}}|4^m N_{ge}\ln(56KN_{ge}/(\epsilon\delta))}{n}\right)\nonumber\\
\leq && 4L_1KC_2 \epsilon +  \xi(\hath)\left(3\sqrt{\frac{|\mathcal{T}_{\mathcal{D}}|4^m N_{ge}\ln(56KN_{ge}/(\epsilon\delta))}{n}} + \frac{2|\mathcal{T}_{\mathcal{D}}|4^m N_{ge}\ln(56KN_{ge}/(\epsilon\delta))}{n}\right), 
\end{eqnarray}
  where $\mathcal{I}_r^{\mathcal{D}}=\{i\in[n]:\bm{z}^{(i)}\in \mathcal{C}_r\}$, $\xi(\hath):=\max_{\bm{z}\in \mathcal{Z}}(\ell(\hath, \bm{z}))$, and $\mathcal{T}_{\mathcal{D}}:=\{r\in [R]:|\mathcal{I}_r^{\mathcal{D}}|\geq 1\}$.  
\end{proof}

\subsection{Proof of Lemma \ref{lem:lemma-rob-cov}}\label{append:subsec:lemma-rob-cov}
The proof uses the following lemma to quantify the covering number of $\mathcal{X}_Q$ whose proof is given in SM~\ref{append:lem:cov-input-space}.

\begin{lemma}\label{lem:cov-input-space}
Following the settings in Eq.~(\ref{append:eqn:QC_thm2}), the covering number of $\mathcal{X}_Q$ in Eq.~(\ref{append:eqn:input_state_space}) is 
  \begin{equation}
    \mathcal{N}(\mathcal{X}_Q,  \epsilon, \|\cdot \|_F) \leq     \left(\frac{28 N_{ge}  }{\epsilon} \right)^{4^mN_{ge}}. 
  \end{equation}
\end{lemma}

\begin{proof}[Proof of Lemma \ref{lem:lemma-rob-cov}]
When QC is applied to accomplish the $K$-class classification task, the sample space is $\mathcal{Z}=\mathcal{X}_Q\times \mathcal{Y}$ with $\mathcal{Y}=\{1,2,...,K\}$. Denote $\tilde{\mathcal{X}}_Q$ as the $\epsilon$-cover set of $\mathcal{X}_Q$ with the covering number $\mathcal{N}(\mathcal{X}_Q, \epsilon, \|\cdot\|_F)$ in Definition \ref{def:cov-num}. Supported by the $\epsilon$-cover set $\tilde{\mathcal{X}}_Q$,  the space $\mathcal{X}_Q\times \{i\}$  can be divided into $\mathcal{N}(\mathcal{X}_Q, \epsilon, \|\cdot\|_F)$ sets for $\forall i\in [K]$. In other words, we can divide $\mathcal{Z}$ into $K\mathcal{N}(\mathcal{X}_Q, \epsilon, \|\cdot\|_F)$ sets denoted by $\{\mathcal{Z}_i\}_{i=1}^{K\mathcal{N}(\mathcal{X}_Q, \epsilon, \|\cdot\|_F)}$.

We then utilize the divided sets of $\mathcal{Z}$  to connect the robustness with covering number according to Definition 1. Given a training example $(\bxi, y^{(i)})$ and a test example $(\bx, y)$, suppose that the corresponding quantum examples $(\sigma(\bxi), y^{(i)})$ and $(\sigma(\bx), y)$ are in the same set of $\{\mathcal{Z}_i\}_{i=1}^{K\mathcal{N}(\mathcal{X}_Q, \epsilon, \|\cdot\|_F)}$. For convenience, we abbreviate $\sigma(\bxi)$ and $\sigma(\bx)$ as $\sigma^{(i)}$ and $\sigma$, respectively. Following the definition of covering number, we have  
\begin{equation}\label{append:eqn:prooflem-1}
y^{(i)} =y~ \text{and}~  \|\sigma^{(i)} - \sigma\|_F\leq 2\epsilon. 
\end{equation}
Since the encoded state takes the form $\sigma=U_E(\bx)(\ket{0}\bra{0})^{\otimes N} U_E(\bx)^{\dagger}$, we have 
\begin{equation}\label{append:eqn:prooflem-2}
  rank(\sigma^{(i)} - \sigma)\leq 2.
\end{equation}
Then, in accordance with the definition of robustness, we bound the discrepancy of the loss values for $\sigma^{(i)}$ and  $\sigma$, i.e.,
\begin{eqnarray}
&& \left|l(h_{Q}(\sigma^{(i)}), y^{(i)}) - l(h_{Q}(\sigma), y )\right| \nonumber\\
\leq && L_1\left\|[\Tr(\mathcal{E}(\sigma^{(i)})o^{(k)})]_{k=1:K}-[\Tr(\mathcal{E}(\sigma))o^{(k)})]_{k=1:K}\right\|_2 \nonumber\\
\leq && L_1 K \max_{k\in K} |\Tr(\mathcal{E}(\sigma^{(i)}))o^{(k)}) - \Tr(\mathcal{E}(\sigma)o^{(k)})| \nonumber\\
\leq && L_1K \max_k \left\|o^{(k)}\right\|_2 \Tr(|\mathcal{E}(\sigma^{(i)} - \sigma)|) \nonumber\\
\leq &&  2L_1K C_2 \|\mathcal{E}\|_{\diamond} \|\sigma^{(i)} - \sigma\|_F \nonumber\\
\leq && 4L_1K C_2 \|\mathcal{E}\|_{\diamond} \epsilon,
\end{eqnarray}
where the first inequality uses the Lipschitz property of the loss function with $\ell(\bm{a},\bm{b}) - \ell(\bm{c},\bm{d})\leq L_1\|\bm{a}-\bm{c}\|_2$ and the form of $\mathcal{E}$ in Lemma \ref{lem:lemma-rob-cov}, the second inequality comes from the definition of $l_2$ norm, the third inequality exploits von Neumann's trace inequality $|\Tr(AB)|\leq \|A\|_p\|B\|_q$ with $1/p+1/q=1$ and the linear property of CPTP map with $\mathcal{E}(\rho) - \mathcal{E}(\sigma) = \mathcal{E}(\rho - \sigma)$, the last second inequality employs $\max_k \left\|o^{(k)}\right\|_2\leq C_2$, the relation $\|\mathcal{E}(\rho-\sigma)\|_1\leq \|\mathcal{E}\|_{\diamond}\|\rho-\sigma\|_1$ and $\|A\|_1\leq rank(A)\|A\|_F$, and the last inequality adopts the result in Eq.~(\ref{append:eqn:prooflem-1}).

The above result exhibits that  the learned QC is  $(K\mathcal{N}(\mathcal{X}_Q, \epsilon, \|\cdot\|), 4L_1KC_2\|\mathcal{E}\|_{\diamond}\epsilon)$-robust. In this regard, the proof can be completed when the upper bound of the covering number $\mathcal{N}(\mathcal{X}_Q, \epsilon, \|\cdot\|_F)$ is known. Supported by Lemma \ref{lem:cov-input-space}, we obtain $\mathcal{N}(\mathcal{X}_Q,  \epsilon, \|\cdot \|_F) \leq  (\frac{28N_{ge}}{\epsilon} )^{4^m N_{ge}}$. Taken together,  the learned QC is \[\left(K \left(\frac{28N_{ge}}{\epsilon} \right)^{4^m N_{ge}}, 4L_1KC_2\|\mathcal{E}\|_{\diamond}\epsilon \right)-robust.\]

\end{proof}

\subsection{Proof of Lemma \ref{lem:cov-input-space}}\label{append:lem:cov-input-space}
The derivation of the covering number of $\mathcal{X}_Q$ in Eq.~(\ref{append:eqn:input_state_space}) uses the following lemma.

\begin{lemma}[Lemma 1, \cite{Barthel2018fundamental}]\label{lem:cov-fund}
For $0<\epsilon<1/10$, the $\epsilon$-covering number for the unitary group $U(2^m)$ with respect to the Frobenius-norm distance in Definition \ref{def:cov-num} obeys
\begin{equation}
  \left(\frac{3}{4\epsilon} \right)^{4^{m}} \leq \mathcal{N}(U(2^m), \epsilon, \|\cdot\|_F) \leq \left(\frac{7}{\epsilon} \right)^{4^{m}}.
\end{equation}  
\end{lemma}

\begin{proof}[Proof of of Lemma \ref{lem:cov-input-space}]  
Recall the input state space is  $\mathcal{X}_{Q}= \{U_E(\bx) (\ket{0}\bra{0})^{\otimes N} U_E(\bx)^{\dagger}| \bx \in \mathcal{X}  \}$, where the encoding unitary $U_E(\bx)=\prod_{g=1}^{N_g}u_g(\bx) \in\mathcal{U}(2^N)$ consists of $N_{ge}$ variational gates and  $N_g-N_{ge}$ fixed gates. To quantify the covering number $\mathcal{N}(\mathcal{X}_Q,  \epsilon, \|\cdot \|_F)$, we define $\tilde{S}$ as the $\epsilon$-covering set for the unitary group $U(2^m)$, $\tilde{\mathcal{X}}_Q$ as the $\epsilon'$-covering set of $\mathcal{X}_Q$, and define a set 
\begin{equation}
  \tilde{\mathcal{U}}_{E}:=\left\{\prod_{i\in\{N_{ge}\}}u_i(\bx)\prod_{j\in\{N_g-N_{ge}\}}u_j(\bx)\Big| u_i(\bx)\in \tilde{S} \right\},
 \end{equation}
 where $u_i(\bm{\theta}_i)$ and $u_j$ specify to the variational and fixed quantum gates, respectively. Note that for any encoding circuit $U_E(\bx)$, we can always find a unitary $U_{E,\epsilon}(\bx)\in\tilde{\mathcal{U}}_{E}$ where each $u_g(\bx)$ is replaced by the nearest element in  the covering set $\tilde{S}$. To this end, following the definition of covering number, the discrepancy between $U_E(\bx) (\ket{0}\bra{0})^{\otimes N} U_E(\bx)^{\dagger}\in \mathcal{X}_Q$ and $U_{E,\epsilon}(\bx) (\ket{0}\bra{0})^{\otimes N} U_{E,\epsilon}(\bx)^{\dagger}\in \tilde{\mathcal{X}}_Q$  under the Frobenius norm satisfies
 \begin{eqnarray}\label{append:eqn:prooflem3-1}
  && \left\|U_E(\bx) (\ket{0}\bra{0})^{\otimes N} U_E(\bx)^{\dagger} - U_{E,\epsilon}(\bx) (\ket{0}\bra{0})^{\otimes N} U_{E,\epsilon}(\bx)^{\dagger}\right\|_F \nonumber\\
  \leq && 2\left\|U_E(\bx) (\ket{0}\bra{0})^{\otimes N} U_E(\bx)^{\dagger} - U_{E,\epsilon}(\bx) (\ket{0}\bra{0})^{\otimes N} U_{E,\epsilon}(\bx)^{\dagger}\right\| \nonumber\\
  \leq &&  2\|U_E(\bx) -U_{E_\epsilon}(\bx)\|   \|(\ket{0}\bra{0})^{\otimes N} \| \nonumber\\
  \leq &&   4 N_{ge} \epsilon,
 \end{eqnarray} 
where the first inequality uses $\|X\|_F\leq rank(X)\|X\|$ and the relation in Eq.~(\ref{append:eqn:prooflem-2}), the second inequality comes from the Cauchy–Schwarz inequality, and the last inequality follows  $ \|U_E(\bx) -U_{E,\epsilon}(\bx)\| \leq N_{ge}\epsilon$ and $\|(\ket{0}\bra{0})^{\otimes N} \|=1$. In other words, $\epsilon'=2 N_{ge} \epsilon$ and  $\tilde{\mathcal{X}}_Q$ is a $(4 N_{ge} \epsilon)$-covering set for $\mathcal{X}_{Q}$. In conjunction with the observation that there are $|\tilde{\mathcal{S}}|^{N_{ge}}$ combinations for the gates in $\tilde{\mathcal{U}}_E$ and the results in Lemma \ref{lem:cov-fund}, we obtain the cardinality of the set $\tilde{\mathcal{U}}_{E}$ is upper bounded by $|\tilde{\mathcal{U}}_{E}|\leq\left(\frac{7}{\epsilon} \right)^{4^m N_{ge}}$. Accordingly, supported by Eq.~(\ref{append:eqn:prooflem3-1}), the covering number of $\mathcal{X}_Q$ satisfies 
\begin{equation}
   \mathcal{N}(\mathcal{X}_Q,  4N_{ge}\epsilon, \|\cdot \|_F)\leq \left(\frac{7}{\epsilon} \right)^{4^m N_{ge}}.
\end{equation}
After simplification, we have 
\begin{equation}
 \mathcal{N}(\mathcal{X}_Q,  \epsilon, \|\cdot \|_F)\leq \left(\frac{28N_{ge}}{\epsilon} \right)^{4^m N_{ge}}.  
\end{equation}
 
\end{proof}

\section{Proof of Corollary 1}\label{append:proof-them4}
The proof leverages the following two lemmas related to the Haar measure and the unitary $t$-design.

\begin{lemma}\label{lem:Tr(WW)}
    Let $\{ W_y\}_{y\in Y} \subset U(d)$ form a unitary $t$-design with $t>1$, and let $A, B: \mathcal{H}_d\to\mathcal{H}_d$ be arbitrary linear operators. Then
    \begin{equation}
        \frac{1}{|Y|}\sum_{y\in Y}\Tr[W_{y}AW_{y}^{\dagger}B]=\int_{\haar}d\mu(W)\Tr[W_{y}AW_{y}^{\dagger}B]=\frac{\Tr[A]\Tr[B]}{d}.
    \end{equation}
\end{lemma}

\begin{lemma}\label{lem:Tr(WW)Tr(WW)}
    Let $\{ W_y\}_{y\in Y} \subset U(d)$ form a unitary $t$-design with $t>1$, and let $A, B, C, D: \mathcal{H}_d\to\mathcal{H}_d$ be arbitrary linear operators. Then
    \begin{align}
        \frac{1}{|Y|}\sum_{y\in Y}\Tr[W_{y}AW_{y}^{\dagger}B]\Tr[W_{y}CW_{y}^{\dagger}D]=&\int_{\haar}d\mu(W)\Tr[W_{y}AW_{y}^{\dagger}B]\Tr[W_{y}CW_{y}^{\dagger}D] \nonumber \\
        =&\frac{1}{d^2-1}\left(\Tr[A]\Tr[B]\Tr[C]\Tr[D] + \Tr[AC]\Tr[BD] \right) \nonumber \\
        & - \frac{1}{d(d^2-1)}\left(\Tr[AC]\Tr[B]\Tr[D] + \Tr[A]\Tr[C]\Tr[BD] \right).
    \end{align}
\end{lemma}

\begin{corollary-non}[Restatement of Corollary 1]\label{append:thm:nogonc}
Following notations in Lemmas 2 and 3, when the  encoding unitary $\{U_E(\bx)|\bx\in \mathcal{X}\}$ forms a 2-design, with probability $1-\delta$, the empirical QC follows  $|	\Tr\left(\sigma(\bxik)\sigma(\bx)\right) - \frac{1}{2^N}| \leq \sqrt{\frac{3}{2^{2N}\delta}}$. When the adopted  ansatz  $\{U(\btheta)|\btheta\in \Theta\}$ forms a 2-design, with probability $1-\delta$, the empirical QC follows $|\Tr(\rhoik o^{(k')}) -\frac{\Tr(o^{(k')})}{2^{D}} |  < \sqrt{\frac{\Tr(o^{(k')})^2 +  2\Tr((o^{(k')})^2)}{2^{2D} \delta }}$.\end{corollary-non}
\begin{proof}[Proof of Corollary 1]
We complete the proof by separately analyzing the concentration behavior  of the encoding unitary and the Ans\"atze. Note that the proof can be adapted from the proof presented in the studies of \cite{mcclean2018barren,cerezo2020cost}. We have included it here for completeness and to establish a clear link between our research and previous work on this subject.  

\medskip
\noindent\underline{\textit{Concentration of the encoding unitary.}} Recall that Condition (iii) in Lemma 2 concerns the distance between two feature states $\rhoik$ and $\rho^{(i',k')}$ for $\forall i,i\in[n_c]$ and $\forall k,k' \in[K]$. In this regard, we quantify the distance between the encoded state $\sigma(\bxik)$ and $\sigma(\bx)$ with $\bx\sim \mathcal{X}$ when the deep encoding ansatz $U_E$ is employed. In particular, we have  
\begin{eqnarray}\label{append:eqn:proof-thm4-exp-ue}
	&& \mathbb{E}_{\bx\sim \mathcal{X}}\left( \Tr\left(\sigma(\bxik)\sigma(\bx) \right) \right) \nonumber\\
	= &&\mathbb{E}_{\bx\sim \mathcal{X}} \left( \Tr\left(\sigma(\bxik)  U_E(\bx)(\ket{0}\bra{0})^{\otimes N}U_E(\bx)^{\dagger} \right) \right) \nonumber\\
	= && \int_{\haar} d\mu(U) \Tr\left(\sigma(\bxik)  U(\ket{0}\bra{0})^{\otimes N}U \right) \nonumber\\
	 = && \frac{\Tr(\sigma(\bxik))\Tr(\ket{0}\bra{0})^{\otimes N})}{2^N} \nonumber\\
	 = && \frac{1}{2^N},
\end{eqnarray}
where the third equality uses Lemma \ref{lem:Tr(WW)}. 
Moreover, the variance of the term $\Tr(\sigma(\bxik)\sigma(\bx))$ yields
\begin{eqnarray}\label{append:eqn:proof-thm4-var-ue}
	&& \text{Var}_{\bx\sim \mathcal{X}}\left( \Tr\left(\sigma(\bxik)\sigma(\bx) \right) \right) \nonumber\\
	= && \mathbb{E}_{\bx\sim \mathcal{X}}\left( \Tr\left(\sigma(\bxik)\sigma(\bx) \right)^2 \right) - \mathbb{E}_{\bx\sim \mathcal{X}}\left( \Tr\left(\sigma(\bxik)\sigma(\bx) \right) \right)^2 \nonumber\\
	 = && \int_{\haar} d\mu(U) \Tr\left(\sigma(\bxik)  U(\ket{0}\bra{0})^{\otimes N}U \right)\Tr\left(\sigma(\bxik)  U(\ket{0}\bra{0})^{\otimes N}U \right) - \frac{1}{2^{2N}} \nonumber\\
	 = && \frac{1}{2^{2N}-1}\left(1 + \Tr(\sigma(\bxik)^2) \right) - \frac{1}{2^{2N}(2^{2N}-1)}\left( \Tr(\sigma(\bxik)^2) + 1\right) - \frac{1}{2^{2N}} \nonumber\\
	 \leq && \frac{1}{2^{2N -2}}  - \frac{1}{2^{2N}} \nonumber\\
	 = && \frac{3}{2^{2N}},
\end{eqnarray}
where the second equality uses the property that the deep encoding  unitary forms 2-design and the result in Eq.~(\ref{append:eqn:proof-thm4-exp-ue}), the third equality comes from Lemma \ref{lem:Tr(WW)}, the  inequality adopts $\Tr(\sigma^2)\leq 1$ and $2^{2N}-1 > 2^{2N-1}$, and the last equality is obtained via simplification. 

Supported by the Chebyshev's inequality $\Pr(|X-\mathbb{E}[X]|\geq a)\leq \text{Var}[X]/a^2$, Eqs.~(\ref{append:eqn:proof-thm4-exp-ue}) and (\ref{append:eqn:proof-thm4-var-ue}) indicate 
 \[ \Pr\left(\Big|\Tr\left(\sigma(\bxik)\sigma(\bx)\right) - \frac{1}{2^N}\Big|\geq \tau \right)\leq \frac{3}{2^{2N} \tau^2}. \]
Equivalently, with probability $1-\delta$, we have 
 \begin{equation}
 \Big|	\Tr\left(\sigma(\bxik)\sigma(\bx)\right) - \frac{1}{2^N}\Big| \leq \sqrt{\frac{3}{2^{2N}\delta}}. 
 \end{equation}

\medskip
\noindent\underline{\textit{Concentration of the deep ansatz.}}  Recall Condition (ii) in Lemma 2. Given a feature state $\rhoik$, for $\forall i\in [n_c]$ and $\forall k\in[K]$ and a measure operator  $o^{(k)}$, the optimal feature state should satisfy \[\Tr(\rho^{*(i,k)} o^{(k')})=\delta_{k,k'}.\] 
In other words, we should examine the value of $\Tr(\rhoik o^{(k')})$ when $\rhoik$ is prepared by a deep \revise{ansatz} $U(\btheta)$. Specifically, we have
\begin{eqnarray}\label{append:eqn:proof-thm4-exp-utheta}
&& \mathbb{E}_{\btheta\sim \Theta}\left(\Tr(\rhoik o^{(k')})\right) \nonumber\\
= && \mathbb{E}_{\btheta\sim \Theta}\left(\Tr(U(\btheta)\sigma(\bxik)U(\btheta)^{\dagger}(o^{(k')}\otimes \mathbb{I}_{2^{N-D}})\right) \nonumber\\
= && \int_{\haar}d\mu(U)\Tr\left( U  \sigma(\bxik)  U^{\dagger} (o^{(k')}\otimes \mathbb{I}_{2^{N-D}}) \right) \nonumber\\
= && \frac{ \Tr(o^{(k')})(2^{N-D})}{2^N}\nonumber\\
= && \frac{\Tr(o^{(k')})}{2^{D}},
\end{eqnarray}
where the first equality comes from the explicit form of QC in Eq.~(4) of the main text, the second equality uses the fact that $U$ follows the Haar distribution, and the last second equality comes from Lemma \ref{lem:Tr(WW)}.

We then quantify the variance of $\Tr(\rhoik o^{(k')})$, i.e.,
\begingroup
\allowdisplaybreaks
\begin{eqnarray}\label{append:eqn:proof-thm4-var-utheta}
    &&\text{Var}_{\btheta\sim \Theta}\left(\Tr(\rhoik o^{(k')})\right) \nonumber\\ 
    = &&  \mathbb{E}_{\btheta\sim \Theta}\left(\Tr(\rhoik o^{(k')})^2\right)  -  \left( \mathbb{E}_{\btheta\sim \Theta} \left( \Tr(\rhoik o^{(k')}) \right) \right)^2 \nonumber\\
    = && \int_{\haar}d\mu(U)\Tr\left(U  \sigma(\bxik)  U^{\dagger} (o^{(k')}\otimes \mathbb{I}_{2^{N-D}})\right)^2 - \frac{\Tr(o^{(k')})^2}{2^{2D}} \nonumber\\
        =&&\frac{1}{2^{2N}-1}\left(\Tr(\sigma(\bxik))\Tr(o^{(k')}\otimes   \mathbb{I}_{2^{N-D}})\Tr(\sigma(\bxik))\Tr(o^{(k')}\otimes   \mathbb{I}_{2^{N-D}}) + \Tr(\sigma(\bxik)^2)\Tr((o^{(k')}\otimes   \mathbb{I}_{2^{N-D}})^2) \right) \nonumber \\
        && - \frac{1}{2^N(2^{2N}-1)}\left(\Tr(\sigma(\bxik)^2)\Tr(o^{(k')}\otimes   \mathbb{I}_{2^{N-D}})^2 + \Tr(\sigma(\bxik))^2\Tr((o^{(k')}\otimes   \mathbb{I}_{2^{N-D}})^2) \right) -  \frac{\Tr(o^{(k')})^2}{2^{2D}} \nonumber\\
        \leq &&\frac{1}{2^{2N}-1}\left(\Tr(o^{(k')}\otimes   \mathbb{I}_{2^{N-D}})^2  +  \Tr((o^{(k')}\otimes   \mathbb{I}_{2^{N-D}})^2) \right)   -  \frac{\Tr(o^{(k')})^2}{2^{2D}}\nonumber\\
       = && \frac{1}{2^{2N}-1}\left(\Tr(o^{(k')})^2 2^{2N-2D}  +  \Tr((o^{(k')})^2) 2^{2N-2D}  \right)  -  \frac{\Tr(o^{(k')})^2}{2^{2D}} \nonumber\\
        \leq  && \frac{ \Tr(o^{(k')})^2 +  \Tr((o^{(k')})^2)   }{2^{2D-1}}  -  \frac{\Tr(o^{(k')})^2}{2^{2D}} \nonumber\\ 
  =  && \frac{ \Tr(o^{(k')})^2 +  2\Tr((o^{(k')})^2)  }{2^{2D}}.
   \end{eqnarray}
\endgroup
 where the second equality uses the fact that $U$ forms the 2-design and Eq.~(\ref{append:eqn:proof-thm4-exp-utheta}), the the third equality comes from Lemma \ref{lem:Tr(WW)Tr(WW)},  the first inequality arises from dropping some positive terms, the last second equality employs $\Tr(A\otimes B)=\Tr(A)\Tr(B)$ and $(A\otimes B)(C\otimes D)=(AC)\otimes (BD)$, and the last inequality exploits $(2^{2N}-1)^{-1}> (2^{N-1})^{-1}$, and the last equalities is obtained via simplification.

Supported by the Chebyshev's inequality $\Pr(|X-\mathbb{E}[X]|\geq a)\leq \text{Var}[X]/a^2$, Eqs.~(\ref{append:eqn:proof-thm4-exp-utheta}) and (\ref{append:eqn:proof-thm4-var-utheta}) indicate 
 \[ \Pr\left(\Big|\Tr(\rhoik o^{(k')}) -\mathbb{E}\left( \Tr(\rhoik o^{(k')}) \right)\Big|\geq \tau \right)\leq \frac{\Tr(o^{(k')})^2 +  2\Tr((o^{(k')})^2)}{2^{2D} \tau^2}. \]
Equivalently, with probability $1-\delta$, we have
 \begin{equation}
  \Big|\Tr(\rhoik o^{(k')}) -\frac{\Tr(o^{(k')})}{2^{D}} \Big|  < \sqrt{\frac{\Tr(o^{(k')})^2 +  2\Tr((o^{(k')})^2)}{2^{2D} \delta }}.
 \end{equation}
 
\end{proof}

\section{More details for the implications of Theorem 1}\label{append:sec:imp-thm1}
 
This section expands the implications of Theorem 1 omitted in the main text. In SM~\ref{append:subsec-ETF}, we elucidate how our results in Theorem 1 relate to EFT and SIC-POVM. Then, in SM \ref{append:subsec-IT}, we provide more explanations about how our results connect with the results in Ref.~\cite{banchi2021generalization}. Next, in SM \ref{append:subsec:gene-bound-discussion}, we interpret why prior generalization bounds become vacuous in the over-parameterized regime. Subsequently, we illustrate how our results complements with convergence theory of quantum neural networks in SM~\ref{append:subsec:conv-QNN}. After, in SM~\ref{append:subsec:expected_risk_approximate}, we discuss the expected risk of QCs when the training loss is near-optimal. Last, in SM~\ref{append:subsec:sum_strategy_QC}, we summarize how the results in Theorem 1 provide insights into the construction strategies of  QCs with the improved performance.

\subsection{Connection with ETF and SIC-POVM}\label{append:subsec-ETF}
 
In this subsection, we explain how the results in Theorem 1 connect with ETF and SIC-POVM. It is noteworthy that the definition of ETF discussed in the context of deep learning and quantum information theory differs, where the former pertains to the case of $2^D \geq K$, while the latter is focused on the setting of $2^D\leq K$.  As our work resides at the intersection of quantum computing and machine learning, the results in Theorem 1 encompass both of these settings. To this end, we begin by presenting the definition of general simplex ETFs utilized in deep learning and elucidating their connections to our findings. Then, we introduce the definition of formal ETFs utilized in quantum information theory and elaborate their connections with our results. Last, we exhibit the connections of SIC-POVM and the results of Theorem 1.  The relationship between the results of Theorem 1 and general simplex ETF, formal ETF, and SIC-POVM is based on the interplay between  the locality of measure operators $D$ and the number of classes $K$, as shown in Fig.~\ref{fig:ETF_SIC}. 
 
 \begin{figure}[h!]
	\centering
\includegraphics[width=0.99\textwidth]{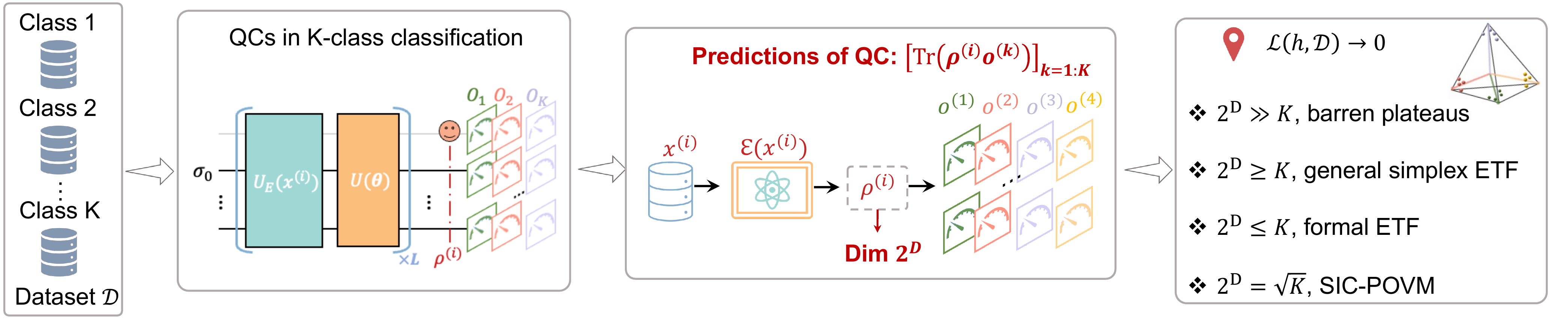}
\caption{\small{\textbf{Connection with general simplex ETFs, formal ETFs, and SIC-POVM}.}}
\label{fig:ETF_SIC}
\end{figure}  
  
\medskip
\noindent\textit{\underline{Connection with general simplex ETF and deep learning}}. The formal definitions of general simplex ETFs is as follows.
\begin{definition}[General simplex ETF, \cite{papyan2020prevalence}]\label{def:NC} 
The standard simplex equiangular tight frame (ETF) is a collection of points in $\mathbb{R}^K$ specified by the columns of $M = \sqrt{\frac{K}{K-1}} (\mathbb{I}_K - \frac{1}{K}\bm{1}_K\bm{1}_K^{\top})$. The general simplex ETF is defined as a collection of points in $\mathbb{R}^{2^D}$ with  $2^D \geq K$ specified by columns of 
\begin{equation}
\tilde{M}\propto \sqrt{\frac{K}{K-1}}P\left(\mathbb{I}_K - \frac{1}{K}\bm{1}_K\bm{1}_K^{\top} \right)	
\end{equation}
with $P\in \mathbb{R}^{2^D\times K}$ is an orthonormal matrix. 
\end{definition}
\noindent Refs.~\cite{papyan2020prevalence,han2022neural} proved that for a deep neural classifier with perfect training, its last-layer features form general simplex ETFs, dubbed \textit{neural collapse}. Suppose that the dimension of the last-layer features is $2^D$ and the number of classes is $K$ with $2^D\geq K$. According to Definition \ref{def:NC}, each class mean corresponds to one column in $\tilde{M}$ and any two class-means $\tilde{M}_{:,i}$ and $\tilde{M}_{:,j}$ are equiangular, i.e.,
\begin{equation}\label{eqn:ETF-DC}
	\langle \tilde{M}_{:,i}, \tilde{M}_{:,j}\rangle = -\frac{1}{K-1},~\forall i, j \in [K].
\end{equation}

We would like to emphasize that Theorem 1 in the main text  provides the same insight into the nature of quantum feature states when a QC reaches perfect training. Specifically, Theorem 1 demonstrates that when $2^D\geq K$, the feature states of a QC form a general simplex ETF, up to a scaling factor. Recall that the first two conditions in Theorem 1 state that the feature states have zero variance within the same class, i.e., 
\[\rho^{(k)}\equiv\rho^{(1,k)}=,..,=\rhoik,...=\rho^{(n_c,k)}\in \mathbb{C}^{2^D\times 2^D}\] with $n_c$ being the number of training examples in each class, and are of equal length and orthogonal in different classes, i.e., 
\[\Tr\Big(\rho^{(k)}\rho^{(k')}\Big)=\delta_{k,k'},~\forall k,k'\in [K].\]
In this respect, the feature states form an orthogonal frame. Given that any orthogonal frame can be transformed into a simplex ETF by scaling down its global mean, we obtain that the feature states of QCs constitute a general simplex ETF, as defined in Definition \ref{def:NC}, when these conditions are met. Denote $\overrightarrow{\rho}^{(k)}$ as the vectorization of $\rho^{(k)}$ for any $k\in[K]$ and the global feature mean as $\overrightarrow{\mu}=\sum_{k=1}^K \rho^{(k)}/K$. The distance for any two quantum feature states after scaling down their global mean is
\begin{eqnarray}\label{eqn:ETF-QC}
	 \Big \langle \overrightarrow{\rho^{(k)}}-\overrightarrow{\mu}, \overrightarrow{\rho^{(k')}} - \overrightarrow{\mu} \Big \rangle   = \Big\langle \overrightarrow{\rho^{(k)}}, \overrightarrow{\rho^{(k')}} \Big\rangle - \Big\langle \overrightarrow{\rho^{(k)}}, \overrightarrow{\mu} \Big\rangle - \Big\langle  \overrightarrow{\rho^{(k')}}, \overrightarrow{\mu} \Big\rangle + \Big\langle \overrightarrow{\mu},  \overrightarrow{\mu} \Big\rangle  = 0 - \frac{1}{K} - \frac{1}{K} + \frac{K}{K^2} = -\frac{1}{K}.
\end{eqnarray}
The combination of Eq.~(\ref{eqn:ETF-DC}) with Eq.~(\ref{eqn:ETF-QC}) suggests that QCs and deep neural classifiers exhibit similar learning behaviors, in which the corresponding features tend to form a general simplex ETF to reach zero training loss. Another interesting property is that the distance between different feature states only depends on the number of classes $K$ but is independent of the feature dimension $2^D$, which contrasts with formal ETFs when $2^D\leq K$, as we will elaborate on later. Besides, Ref.~\cite{montanaro2008lower} proves that the lower bound of the probability error in quantum state discrimination is
\begin{equation}\label{eqn:error_dist_state}
	\frac{(K-1)F^m}{2K} \leq  P_E \leq 2\sum_{k'> k }\sqrt{p_kp_k'}\sqrt{F(\rho^{(k)},\rho^{(k')})},
\end{equation}
where $K$ refers to the number of quantum states, $m$ denotes the number of measurements, and $F$ is the lower-bound fidelity of any pair of states, i.e., $F(\rho^{(k)},\rho^{(k')})\geq F$ for $\forall k,k'\in [K]$. The orthogonality of features states suggests $P_E =0$ for the optimal QCs. However, we will show that this is not the case when $2^D<K$.

\medskip
\noindent\textit{\underline{Connection with formal ETFs}}.  In the following, we expand the discussions about how the optimal QCs relate to formal ETF and SIC-POVM when $2^D \leq  K$. We acknowledge that exploring the setting of $2^D \leq  K$ is more of a theoretical interest rather than a practical one. It is because in most datasets, the number of classes is limited, and it is natural to set $2^D \geq  K$ for efficient learning.

Let us first recall the findings of Theorem 1, where the perfect  training of QCs can be attained by forming an orthogonal frame using either feature states or measurement operators. This orthogonality is maintained when $2^D\geq K$. We note that when $K=2^D$, our results are in line with both the general simplex ETF in Definition \ref{def:NC} and the formal ETF expressed below.
\begin{definition}[Formal ETF, \cite{strohmer2003grassmannian,sustik2007existence}]\label{def:ETF}
	Let $M$ be a $2^D\times K$ matrix  whose columns are  $M_1,...,M_K$ with $2^D\leq K$. The matrix $M$ is called an equiangular tight frame (ETF) if it satisfies three conditions.
	\begin{enumerate}  
		\item Each column has unit norm with $\|M_{:,i}\|_2=1$ for $\forall i \in [K]$. 
		\item The columns are equiangular. For some nonnegative $\alpha$, we have 
		$|\langle M_{:,i}, M_{:,j} \rangle |=\alpha$ for $i\neq j$.
		\item The columns form a tight frame. That is, $MM^{\dagger}  = (K/2^D)\mathbb{I}_{2^D\times 2^D}$.
	\end{enumerate}
\end{definition}
\noindent According to the above definition, an immediate observation is that when $K = 2^D$, the feature states enabling perfect training, i.e., Conditions (i)\&(ii) in Theorem 1, form a formal ETF in which the absolute inner product between distinct feature vectors is zero, i.e., $\alpha=\Tr(\rho^{(k)}\rho^{(k')})=0$.

\medskip 
We now turn our attention to comprehending the optimal QCs in the case of $2^D<K$. Specifically, in the task of $K$-class classification, define the set of measurement operators $\bm{o}=\{o^{(k)}\}$ as
\[
	o^{(k)} = M_{:,k}M_{:,k}^{\dagger}\in \mathbb{C}^{2^D\times 2^D},~\forall k \in [K],\]
where $M_{:,k}$ refers to the $k$-th column of ETF $M$ in Definition \ref{def:ETF}. The loss function to be minimized is  $\mathcal{L} = \frac{1}{n} \sum_{i=1}^{n_c}\sum_{k=1}^K \|\left[\Tr(\rhoik o^{(k)})\right]_{k=1:K} - \yik \|_2^2,$
 where the label vector is rewritten as $\yik=M_{:,k}$. In this setting, it is easy to extend  the proof techniques of Theorem 1 to show that when perfect training of QCs happens with $\mathcal{L}=0$, the feature states are pure states and have the vanished variability in the same class. Moreover, all feature states are equal length and form a formal ETF $M$. Supported by the results in Ref.~\cite{sustik2007existence}, for any two feature states from the varied classes, we have   
 \begin{equation}\label{eqn:ETF-QC-dist}
 	\Big|\Big \langle \overrightarrow{\rho^{(k)}}, \overrightarrow{\rho^{(k')}}\Big \rangle\Big| =  \sqrt{\frac{K-2^D}{2^D(K-1)}},~\forall k, k'\in[K].
 \end{equation} 
In conjunction with Eq.~(\ref{eqn:error_dist_state}) and Eq.~(\ref{eqn:ETF-QC-dist}), we obtain that the feature states of optimal QCs are indistinguishable, where the error probability is lower bounded by
\begin{equation}
	\frac{(K-1)}{2K}\left(\frac{K-2^D}{2^D(K-1)}\right)^{m/2} \leq  P_E.
\end{equation}  

The above result contradicts the learning dynamics of optimal QCs in the setting of $2^D \geq K$ and $2^D<K$, where the former can achieve the zero error probability but the latter cannot. Moreover, when $2^D \geq K$, the general simplex ETF always exists for any $D$, while numerical results show that formal ETFs arise for very few pairs $(2^D, K)$ \cite{sustik2007existence}. Besides, although achieving zero error probability in discriminating different feature states is unattainable, the optimal QCs in the case of  $2^D<K$ can still achieve a lower bound of error probability that is smaller than that of imperfect QCs with non-zero loss. This is due to the fact that when the feature states do not form an ETF, they are not maximally distant from each other, and the quantify $F$ in Eq.~(\ref{eqn:error_dist_state}) would increase, resulting in a higher lower bound.

\medskip 
\noindent\textit{\underline{Connection with SIC-POVM}}. The aforementioned results can be effectively extended to demonstrate the association between the optimal QCs with zero training loss and SIC-POVM, as the latter is a special case of ETF with $2^D=\sqrt{K}$. Consequently, when SIC-POVM is applied, the feature states of the optimal QCs are pure states and equal length,  have the vanished variability in the same class, and form an ETF with $M\in \mathbb{C}^{\sqrt{K}\times K}$. For any two features from the varied classes, their distance is
\begin{equation}
	\Big|\Big \langle \overrightarrow{\rho^{(k)}}, \overrightarrow{\rho^{(k')}}\Big \rangle\Big| =  \sqrt{\frac{K-\sqrt{K}}{\sqrt{K}(K-1)}},~\forall k, k'\in[K].
\end{equation}  
Moreover, the optimal QCs achieves a lower bound of error probability in discriminating feature states that is smaller than that of imperfect QCs with non-zero loss. The corresponding lower bound is  
\begin{equation}
	\frac{K-1}{2K}\left(\frac{K-\sqrt{K}}{\sqrt{K}(K-1)} \right)^{m/2}\leq P_E,
\end{equation}
where $m$ is the number of measurements.

\subsection{Connection with generalization of QML in the view of information theory}\label{append:subsec-IT}

In this subsection, we detail the intrinsic connection and  difference between our work and Ref.~\cite{banchi2021generalization} omitted in the main text. 

Both of our work and Ref.~\cite{banchi2021generalization} achieve the similar results with respect to the variability of feature states, despite a slight difference in the choice of loss functions. Namely, Ref.~\cite{banchi2021generalization} derives its results under the linear loss, while we consider the mean-square loss with an optional regularization term. However, the overall conclusions and findings remain consistent between the two studies. Specifically, in Ref.~\cite[Eq.~(18)]{banchi2021generalization}, it was pointed out that a low training error is achievable for binary classification tasks when the fidelity between two embedded states is small if the inputs are from different classes and high if the inputs are from the same class. For multi-classification tasks, Ref.~\cite[Appendix A.3]{banchi2021generalization} extended this result by showing that the zero training loss can be achieved when embedded states are almost constant within the same class and orthogonal with those of other classes. These properties echo with the geometric interpretation of feature states achieved in Theorem 1. In particular, the obtained Conditions (i)-(ii) of Theorem 1 extend this argument to more general and practical settings, i.e., the prior information about how to construct the optimal measurement is unnecessary and the variational ansatz $U(\btheta)$ is considered. Recall that QC in Eq.~(4) of the main text takes the form
\begin{equation}\label{eqn:QC-gene}
	[h(\rhoik, o^{(1)}),..., h(\rhoik, o^{(k)}),..., h(\rhoik, o^{(K)})],~\text{and}~h(\rhoik, o^{(k)})=\Tr(\rhoik o^{(k)}) ~\forall k\in[K],
\end{equation}  
where $\{o^{(k)}\}$ is a set of  measure operators and $\rhoik$ is the feature state of the $i$-th example in the $k$-th class with $\rhoik=\Tr_D(U(\btheta)\sigma(\bxik)U(\btheta)^{\dagger})$ and $\sigma(\bxik)$ being the embedded state. In this regard,  QCs exploited in Ref.~\cite{banchi2021generalization} are special cases in Eq.~(\ref{eqn:QC-gene}) such that $U(\btheta)$ corresponds to $\mathbb{I}$ and $\{o^{(k)}\}$ corresponds to the optimal measurements.

We next discuss the connections between our work and Ref.~\cite{banchi2021generalization} in the view of generalization error. The authors in Ref.~\cite{banchi2021generalization} prove that with probability $1-\delta$, the generalization bound of QCs yields
\begin{equation}\label{eqn:gene_ban}
	\mathsf{R}_{\text{Gene}}\leq \sqrt{\frac{\mathcal{B}}{n}} + \sqrt{\frac{2\log(1/\delta)}{n}},
\end{equation}
where $n$ is the number of training examples and $\mathcal{B}$  equals to $2^{I_2(X:Q)}$. In  Ref.~\cite[Section IV.A]{banchi2021generalization}, the authors further use information theory to quantify how different embedding methods affect $I_2(X:Q)$ and show that low-entropy datasets and low-dimensional embeddings lead to a smaller $I_2(X:Q)$ and therefor a lower generalization error.   The generalization bound in Eq.~(\ref{eqn:gene_ban}) and the derived generalization bound in Lemma \ref{thm:gene_robust_QNN}  are consistent in terms of the number of encoding gates, where an increase in the number of encoding gates $N_{ge}$ leads to an increased error bound. This is because when the encoding unitary forms a 2-design,  the states from the same class are orthogonal and thus maximize $I_2(X:Q)$, implying a large $\mathcal{B}$.

An attractive feature of our generalization error bound over that in Ref.~\cite{banchi2021generalization} is capturing how the training loss dynamically effects the generalization ability of QCs. This behavior is reflected by the term $|\mathcal{T}|$, which can be decreased from $n$ to $K$ during the optimization. Besides, the generalization bound achieved in our work is more `practical', since it provides a more intuitive description of its dependence on the training loss and the number of encoding gates.

\subsection{Non-vacuous generalization bound in the over-parameterized regime}\label{append:subsec:gene-bound-discussion}
 
In this subsection, we first elucidate why prior results related to the generalization of quantum neural networks such as Refs.~\cite{caro2021generalization,du2022efficient,gyurik2021structural,cai2022sample}, cannot fully account for the generalization ability of over-parameterized QCs when   $n\geq N_t$ given in Definition 1 of the main text. Then, we explain why the generalization error bound derived in Lemma \ref{thm:gene_robust_QNN}  can be applied to the over-parameterized regime.

A crucial reason why prior results fail to explain the generalization ability of over-parameterized QCs is their reliance on the fundamental learning-theoretic technique of uniform convergence  \cite{mohri2018foundations}. The approach taken by these results comprises of two steps: (i) quantify the expressivity of the hypothesis space of QNNs using a complexity measure such as Rademacher complexity, VC dimension, or covering number; (ii) use uniform convergence to estimate the generalization error based on the measured expressivity. However, the expressivity of QNNs grows exponentially with the number of parameters, underlying that the sample size must scale polynomially with the number of parameters  $N_t$. For example, the expressivity and the generalization bound achieved in Refs.~\cite{caro2021generalization} are $O(\exp{(N_t)})$ and $O(\sqrt{N_t/n})$, where  $n$ denotes the number of training examples. The explicitly polynomial dependence on the number of trainable parameters $N_t$ results into a vacuous generalization bound  in the over-parameterized regime with $N_t \gg n$.

The issue of the vacuous generalization bound has also arisen when attempting to explain the generalization ability of over-parameterized deep neural networks. Both empirical and theoretical studies in the field of deep learning theory have highlighted that algorithm-independent and expressivity-induced bounds fail to fully account for the generalization ability of such models \cite{zhang2021understanding,nagarajan2019uniform,berner2021modern}, and one possible solution is to shift the focus from quantifying the expressivity of the entire hypothesis set (i.e., the range of the learning algorithm) to the optimized model. In other words,  an algorithm-dependent generalization bound is needed in order to provide a non-vacuous bound in the over-parameterized regime.
 
Enlightened by the progress in deep learning theory, we leverage  the concept of the algorithmic robustness  in Definition \ref{def:robustness} instead of the expressivity-induced approaches to analyze the generalization ability of QCs. Different from the expressivity-induced bound, algorithmic robustness focuses on the learned model $h_{\mathcal{A}_{\mathcal{D}}}\in \mathcal{H}$ rather than quantifying the whole hypothesis space $\mathcal{H}$ represented by QCs. This ensures the obtained generalization bound in Lemma \ref{thm:gene_robust_QNN} is algorithmic-dependent and is non-vacuous in the over-parameterized regime, i.e., with probability $1-\delta$, the generalization error of the learned QC $\hath_Q$ yields   
 \[\RGENE(\hath_Q)\leq O\left(\sqrt{\frac{|\mathcal{T}_{\mathcal{D}}|4^m N_{ge}\ln \frac{56KN_{ge}}{\epsilon\delta}}{n}}\right),\]
 where $N_{ge}$ is the number of encoding gates, $m$ is the maximum number of qubits that the encoding gate can be applied to, $\epsilon$ is a predefined tolerable error, $n$ refers to the number of training examples, $\mathcal{I}_r^{\mathcal{D}}=\{i\in[n]:\bm{z}^{(i)}\in \mathcal{C}_r\}$, $\xi(\hath):=\max_{\bm{z}\in \mathcal{Z}}(\ell(\hath, \bm{z}))$, and $\mathcal{T}_{\mathcal{D}}:=\{r\in [R]:|\mathcal{I}_r^{\mathcal{D}}|\geq 1\}$. 
 
 Compared with prior expressivity-induced generalization bounds scaling with $O(\sqrt{N_t/n})$, the derived bound in our work  does not explicitly depend on the number of trainable parameters $N_t$. Moreover, the derived bound is algorithmic-dependent in which a lower training loss suggests a lower $|\mathcal{T}_{\mathcal{D}}|$, which in turn leads to a better generalization ability. In the optimal case with zero training loss, we have $|\mathcal{T}_{\mathcal{D}}|=K$. The above features allow that the derived bound in Lemma \ref{thm:gene_robust_QNN} can be used to interpret the over-parameterized QCs.

\begin{figure}[h!]
	\centering
\includegraphics[width=0.65\textwidth]{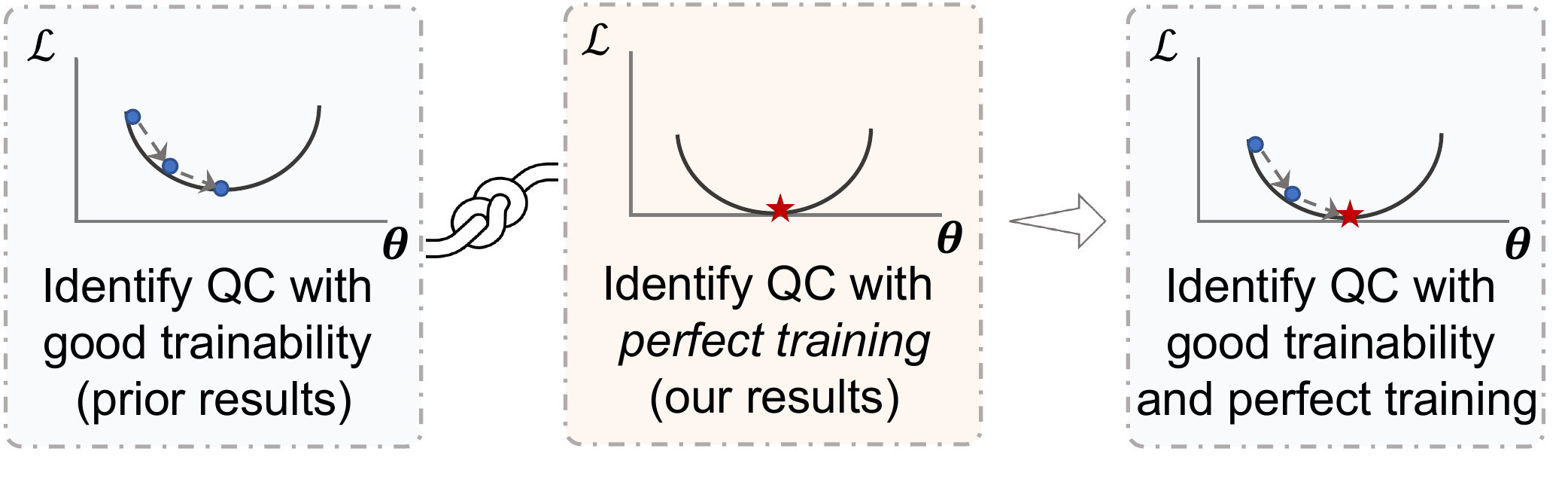}
	\caption{\small{\textbf{The complementary role of analyzing the ability of convergence and the ability to the perfect training}. The left panel indicates the research focus related to barren plateaus, which concerns whether the training parameters can converge to the stationary point, e.g., local and global minima. The middle panel illustrates the research focus of perfect training, which concerns whether the optimal classifier (or equivalently the global minima) can reach the zero loss. The right panel visualizes the complementary role between convergence guarantee and perfect training.}}
	\label{fig:trainability-vs-global-opt}
\end{figure}

\subsection{Relation with convergence theory of quantum neural networks}\label{append:subsec:conv-QNN}
 
We present how the results achieved in Theorem 1 relate to the research on the convergence of QCs in terms of $\RERM$. Note that the loss function $\mathcal{L}$ in Eq.~(\ref{append:eqn:loss_mse_reg_no}) manipulates the ultimate performance of QCs, as it guides the optimization process. Recognizing its significance, considerable effort has been dedicated to understanding the capabilities of QCs by analyzing the properties of the loss function. A crucial research direction in this area is investigating the trainability of QCs, especially for their ability of converging to local or global minima. Theoretical results have shown that improper choices of the encoding unitary $U_E$, ansatz $U(\btheta)$, and measurement operator $\bm{O}$ can result in vanished gradients (a.k.a, barren plateaus), leading to poor trainability of QCs \cite{mcclean2018barren,cerezo2020cost,wang2020noise,holmes2021connecting,marrero2021entanglement,thanasilp2021subtleties,arrasmith2022equivalence}. In addition, some studies have developed advanced tools to diagnose and avoid barren plateaus \cite{grant2019initialization,larocca2022diagnosing,zhang2022escaping,sack2022avoiding}.

In contrast to prior studies that focus on exploring the convergence of QCs, our work aims to comprehensively understand the ultimate performance of QCs by analyzing the loss function in Eq.~(\ref{append:eqn:loss_mse_reg_no}). That is, the conditions under which the optimal parameters of QCs lead to zero loss (a.k.a, perfect training), i.e., $\min_{\btheta} \mathcal{L}(\btheta)\rightarrow 0$. Remarkably, our findings are complementary to the existing research on the convergence ability of QCs. As illustrated in Fig.~\ref{fig:trainability-vs-global-opt}, merely understanding the convergence    of QCs is insufficient to warrant the practical utility of QCs, because the optimized  parameters may still result in a large loss and a large classification error. More specifically, although perfect training is not the necessary condition for the perfect classification, converging to a very large training loss may imply a high classification error, at least in the worst case. However, by combining our findings with the results on the trainability of QCs, we can recognize a class of QCs that can be optimized to achieve the optimal parameters, suggesting zero loss and perfect classification.

\subsection{Expected risk of QCs in the case of approximate satisfaction}\label{append:subsec:expected_risk_approximate}

Here we discuss the expected risk $\ROPT$ of near-optimal QCs, where the empirical risk equals to a small value with $\RERM = \varepsilon$. Recall $\ROPT=\RERM + \RGENE=\varepsilon +\RGENE$. In other words, to quantify $\ROPT$ of near-optimal QCs, it is necessary to comprehend how the imperfect training, or equivalently $\varepsilon$, effects $\RGENE$. 

The derived generalization error bound in Lemma \ref{thm:gene_robust_QNN} offers a straightforward solution to this problem. In particular,  with probability $1-\delta$, the generalization error of the learned QC $\hath_Q$ yields   
\begin{equation}\label{eqn:gene_simp}
\RGENE(\hath_Q)\leq O\left(\sqrt{\frac{|\mathcal{T}_{\mathcal{D}}|4^m N_{ge}\ln \frac{56KN_{ge}}{\epsilon\delta}}{n}}\right),
\end{equation}
where $N_{ge}$ is the number of encoding gates, $m$ is the maximum number of qubits that the encoding gate can be applied to, $\epsilon$ is a predefined tolerable error, $n$ refers to the number of training examples, $\mathcal{I}_r^{\mathcal{D}}=\{i\in[n]:\bm{z}^{(i)}\in \mathcal{C}_r\}$, and $\mathcal{T}_{\mathcal{D}}:=\{r\in [R]:|\mathcal{I}_r^{\mathcal{D}}|\geq 1\}$. The notation $\mathcal{C}_r$ originates from the concept of robustness in Definition \ref{def:robustness}. 
\noindent The relation between $\RGENE$ and $\varepsilon$ is as follows. 
\begin{itemize}
	\item For optimal QCs with $\varepsilon=0$, the vanished variability of feature states in the same class implies that all training examples belong to  $K$ elements of $\{\mathcal{C}_r\}_{r=1}^R$. Consequently, according to the definitions of $\mathcal{I}_r^{\mathcal{D}}$ and  $\mathcal{T}_{\mathcal{D}}$  in Eq.~(\ref{eqn:gene_simp}), we have $|\mathcal{T}_{\mathcal{D}}|=K$. In this case, with probability $1-\delta$, the generalization error yields $\RGENE(\hath_Q)\leq O\Big(\sqrt{(K 4^m N_{ge}\ln \frac{56KN_{ge}}{\epsilon\delta})/n}\Big)$.  
	\item For near-optimal QCs with a very small $\varepsilon$, where all training examples $\{(\bxi, y^{(i)})\}_{i=1}^n$ belong to $K$ elements of the disjoint sets $\{\mathcal{C}_r\}_{r=1}^R$, the generalization error bound is identical to the optimal case, i.e., $\RGENE(\hath_Q)\leq O\Big(\sqrt{(K 4^m N_{ge}\ln \frac{56KN_{ge}}{\epsilon\delta})/n}\Big)$.
	\item For near-optimal QCs with $\varepsilon$ surpassing a threshold, where the feature states from the same class are not sufficiently close and the all training examples $\{(\bxi, y^{(i)})\}_{i=1}^n$ belong to $K'$ elements of the disjoint sets $\{\mathcal{C}_r\}_{r=1}^R$ with $K'>K$, the generalization error bound becomes $\RGENE(\hath_Q)\leq O\Big(\sqrt{(K' 4^m N_{ge}\ln \frac{56KN_{ge}}{\epsilon\delta})/n}\Big)$.    
\end{itemize}
The above analysis conveys the following implications related to QCs in the realistic scenario. When the empirical risk $\varepsilon$ is below a problem-dependent threshold, the generalization error bound is the same with the optimal case. In other words, when the number of training examples satisfies $n\gg O\Big(\sqrt{K 4^m N_{ge}\ln \frac{56KN_{ge}}{\epsilon\delta}}\Big)$, $\RGENE\rightarrow 0$ and the expect risk is $\ROPT=\varepsilon$.  However, when the empirical risk $\varepsilon$ is above such a threshold, a larger number of training examples, i.e., $n\gg O\Big(\sqrt{ K' 4^m N_{ge}\ln \frac{56KN_{ge}}{\epsilon\delta}}\Big)$,  is required to achieve $\RGENE\rightarrow 0$ and $\ROPT=\varepsilon$. Since $K'$ is proportional to $\varepsilon$, it can be concluded that imperfect training hinders the generalization ability of QCs, which requires a larger number of training examples to suppress the generalization error and expected risk.

\subsection{Summary of strategies to construct QCs}\label{append:subsec:sum_strategy_QC}
 
In the main text, we have elucidated how our theoretical results provide insights into the construction strategies of  QCs with improved performance. To ensure clarity, we now summarize these construction strategies.

\begin{itemize}
	\item In order to achieve the optimal power, it is crucial for a QC to have the capability of forming the feature states as ETF. This principle aligns with the principles of quantum metric learning and quantum self-supervised learning \cite{lloyd2020quantum,nghiem2021unified,larose2020robust,jaderberg2022quantum,Yang2022Analog}. Instead of employing  the fixed encoding circuit $U_E$ and ansatz $U(\btheta)$ and adjusting the parameters $\btheta$ to maximize training accuracy, a new strategy is first completing a pretext task to design an encoding circuit and ansatz that can align the feature states with an ETF.  Subsequently, the learned encoding unitary and ansatz are utilized to perform the classification task.
	\item When the number of classes $K$ is not excessively large, employing Pauli-based measurements with $2^D=K$ is advantageous for performing the multi-class classification task. This preference arises due to the effectiveness of Pauli-based measurements in real quantum systems, which allows for the utilization of classical shadow techniques to expedite evaluation. Furthermore, Pauli-based measurements exhibit relative insensitivity to barren plateaus issues compared to computational basis measurements. However, in the case of an extremely large number of classes $K$, the use of SIC-POVM becomes desirable for achieving relatively good performance. 
	\item The training efficiency of QCs can be enhanced by reducing the number of training examples. To determine an appropriate number of training examples, the derived generalization bound in Theorem 1 can serve as a valuable guideline.
\end{itemize}

\section{Implementation of the algorithm to enhance the power of QCs}\label{append:sec:alg-imp}

In practical scenarios, there are many flexible hyper-parameter settings to initiate QCs, each leading to distinct learning performance.  An interpretation is shown in the left panel of Fig.~\ref{fig:relation-with-geometric-QNN}. Namely, for each ansatz, although the structure of the gates in each layer is fixed, the number of layers $L$ can be adjusted to form shallow or deep circuits, implying the employed ansatz can either be under-parameterized or over-parameterized.  To this end, it is desired to estimate the optimal $L^*$ whose minimum expected risk is lower than other settings of $L$, e.g., $U_3$ in Fig.~\ref{fig:relation-with-geometric-QNN}.

\begin{figure}[h!]
\centering
	\includegraphics[width=0.99\textwidth]{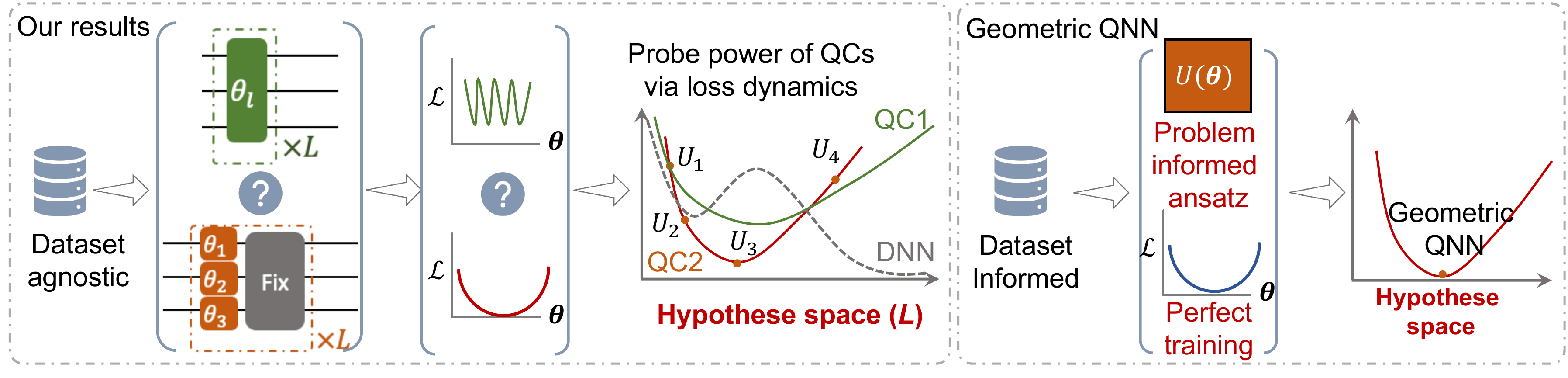}
	\caption{\small{\textbf{Complementary relation of our results and geometric QNNs}. The left and right panels exhibit the achieved results in our work and in Ref.~\cite{schatzki2022theoretical}, respectively.}}
	\label{fig:relation-with-geometric-QNN}
\end{figure}

The above task can be achieved through the following two steps: (1) estimate the optimal hyper-parameter $L^*$, where the corresponding parameter space $\Theta$ (e.g., $\Theta = [0, 2\pi)^{L^*}$ for the green ansatz and $\Theta = [0, 2\pi)^{3L^*}$ for the red ansatz) includes a set of parameters $\btheta^*\in \Theta$ that enable the best risk performance compared to other values of $L$; (2) construct the optimal ansatz $U(\btheta)$ with the layer number $L^*$ and optimizing this ansatz to obtain the optimal parameters $\btheta^*$.  The proposed algorithm orients to leverage the derived U-shaped curve  to identify  $L^*$ as detailed below.

\subsection{Implementation details of the proposed algorithm}
The derived U-shaped curve of QCs indicates that the minimum risk of QC locates at the modest size of the hypothesis space $\mathcal{H}_Q$. In other words, the number of trainable parameters $N_t$ should be lower than $O(poly(N))$, with $N$ being the number of qubits in QC. Moreover, Lemma \ref{thm:gene_robust_QNN} hints that the generalization error of QCs can be well suppressed by using the modest number of train examples. As such, if the available number of training examples in $\mathcal{D}$ is tremendous, we can distill a subset from $\mathcal{D}$ to enhance the training efficiency without increasing the generalization error.

\begin{algorithm}\label{alg:dymanic-loss-curve}
\caption{Estimate risk curves of quantum and classical classifiers}
\KwData{The train dataset $\mathcal{D}$, the test dataset $\mathcal{D}_{Test}$, QC $h_Q$ associated with the hypothesis space $\mathcal{H}_Q$, CC $h_C$ associated with the hypothesis space $\mathcal{H}_Q$, the loss function $\mathcal{L}(\cdot, \cdot)$.}
\KwResult{The estimated risk curves of QC and CC.}
 Initialization: $W$ tuples of hyper-parameter settings $\{n^{(w)}, N_t^{(w)}, T^{(w)}\}_{w=1}^W$ with $n$ being train examples, $N_t$ being the number of trainable parameters, and $T$ being the number of epochs\;
 \For{$w=1$, $w\leq W$, $w++$}{
  Initialize train data as $\mathcal{D}^{(w)}$ by distilling $n^{(w)}$ examples from $\mathcal{D}$\;
  \textcolor{gray}{\# Collect loss dynamics of QC \;}
  Minimize the loss function $\mathcal{L}(\cdot, \cdot)$ via gradient descent methods to obtain the empirical quantum classifier $\bar{h}_Q^{(w)}\in\mathcal{H}_Q$ using $\mathcal{D}^{(w)}$ within $T^{(w)}$ epochs and $N_t$  trainable parameters\;
  Record the loss value $\mathcal{L}(\bar{h}_Q^{(w)},\mathcal{D}_{Test})$ \;
   \textcolor{gray}{\# Collect loss dynamics of CC \;}
    Minimize the loss function $\mathcal{L}(\cdot, \cdot)$ via gradient descent methods to obtain the empirical classical classifier $\bar{h}_C^{(w)}\in\mathcal{H}_C$ using $\mathcal{D}^{(w)}$ within $T^{(w)}$ epochs and $N_t$ trainable parameters\;
  Record the loss value $\mathcal{L}(\bar{h}_C^{(w)},\mathcal{D}_{Test})$ \;  
 }
Fitting the loss dynamics of $\{\mathcal{L}(\bar{h}_Q^{(w)},\mathcal{D}_{Test})\}_{w=1}^W$ to obtain the estimated risk curve of QC \;
Fitting the loss dynamics of $\{\mathcal{L}(\bar{h}_C^{(w)},\mathcal{D}_{Test})\}_{w=1}^W$ to obtain the estimated risk curve of CC.
\end{algorithm}

The Pseudo code of the proposed method is presented in Alg.~\ref{alg:dymanic-loss-curve}. Note that the learning rate, the adopted optimizer, and the batch size can be varied of different classifiers to better estimate the empirical hypothesis. To ensure that the collected results of QC span its basin of the risk curve, the employed $W$ settings of $N_t$ can be acquired by uniformly interpolating from $O(1)$ to $O(poly(N))$. The iteration $T$ should ensure the convergence of QC. Once the loss values of QC and CC under  $\{n^{(w)}, N_t^{(w)}, T^{(w)}\}_{w=1}^W$ are obtained, we can apply certain fitting algorithms to attain their risk curves.  

Moreover, we would like to point out that the collected losses may not necessarily be optimal in practice, as its performance could be influenced by the choice of optimizer and initial parameters. Consequently, this may lead to a lower precision in the fitting curve and an underestimation of the true power of the QC on the given dataset. To address this issue, the use of advanced optimizers and initialization strategies can help to a more accurate identification of QCs.

\subsection{Relation with geometric quantum machine learning}

We now explain the complementary relation between Alg.~\ref{alg:dymanic-loss-curve} and geometric quantum machine learning  \cite{schatzki2022theoretical} in solving classification tasks.  As shown in the left dashed box of Fig.~\ref{fig:relation-with-geometric-QNN}, our work considers an agnostic setting where the prior information of the dataset is unknown, making it impossible to design an effective problem-informed ansatz that guarantees perfect training. Therefore, we analyze how the expected risk changes when the employed ansatz is constructed from shallow to deep, or equivalently, how it evolves as the hypothesis space continuously expands. The resulting U-shaped  risk curve motivates us to devise Alg.~\ref{alg:dymanic-loss-curve} that can locate a suboptimal ansatz to complete the learning task, labeled as `$U_3$' in Fig.~\ref{fig:relation-with-geometric-QNN}. 

However, as depicted in the right dashed box in Fig.~\ref{fig:relation-with-geometric-QNN}, Ref.~\cite{schatzki2022theoretical} and other works related to geometric quantum machine learning \cite{meyer2023exploiting, ragone2022representation,nguyen2022theory} rely on a setting where the prior information of the dataset is available. In such a setting, a problem-informed ansatz can be designed to ensure perfect training and the expected risk can approach to zero when the number of training data is sufficient.

The complementary relation between our results and Ref.~\cite{schatzki2022theoretical} provides the following insights when using QCs to solve classification tasks. On the one hand, the three conditions derived in Theorem 1 raise the question of how to satisfy them. Notably, an interesting observation is that the results achieved in Ref.~\cite{schatzki2022theoretical} offers a readily available solution to address this issue. On the other hand,   our results hint distinct philosophy in the design of quantum machine learning and deep learning models. In the context of deep learning, both over-parameterization (i.e., having many more parameters than training data) and injecting prior information into the neural network design can enhance model performance. However, this is not the case for quantum machine learning. Our result, connected with Ref.~\cite{schatzki2022theoretical}, suggests that designing problem-informed ansatz is a more promising approach than over-parameterization for improving the performance of QCs.

\section{Numerical simulation details}\label{append:sec:sim-res}

\textbf{Dataset.} The construction of the parity dataset mainly follows from Ref.~\cite{cross2015quantum}. Note that this task has also been broadly studied in the field of deep learning to show the limits of deep neural classifiers \cite{daniely2020learning,barak2022hidden}. The constructed dataset contains in total $64$ examples. Each example corresponds to a bit-string with the length $6$, i.e., $\bx\in \{0, 1\}^6$. The label of $\bx$ is assigned to be $1$ if the number of `0' in $\bx$ is even; otherwise, the label is 0. We split it into train dataset and test dataset with the train-test-split ratio being $0.75$. The number of train examples in each class is controlled to be the same. For each example, its feature dimension is $10$. The image dataset is adapted from Ref.~\cite{xiao2017online}. Specifically, the data from the first nine classes are preserved and the total number of examples is $180$. The train-test-split ratio is set as $0.5$ to construct the train and test dataset. Each example corresponds to an image with $28\times 28$ pixels. In the preprocessing stage, we flatten all examples followed by padding and normalization. The processed example yields an $10$-qubit state with $\bx\in \mathbb{R}^{2^{10}}$ and $\|\bx\|_2^2=1$. Some examples after preprocessing are illustrated in Fig.~\ref{fig:append:dataset}(a).

\begin{figure*}
	\centering
	\includegraphics[width=0.92\textwidth]{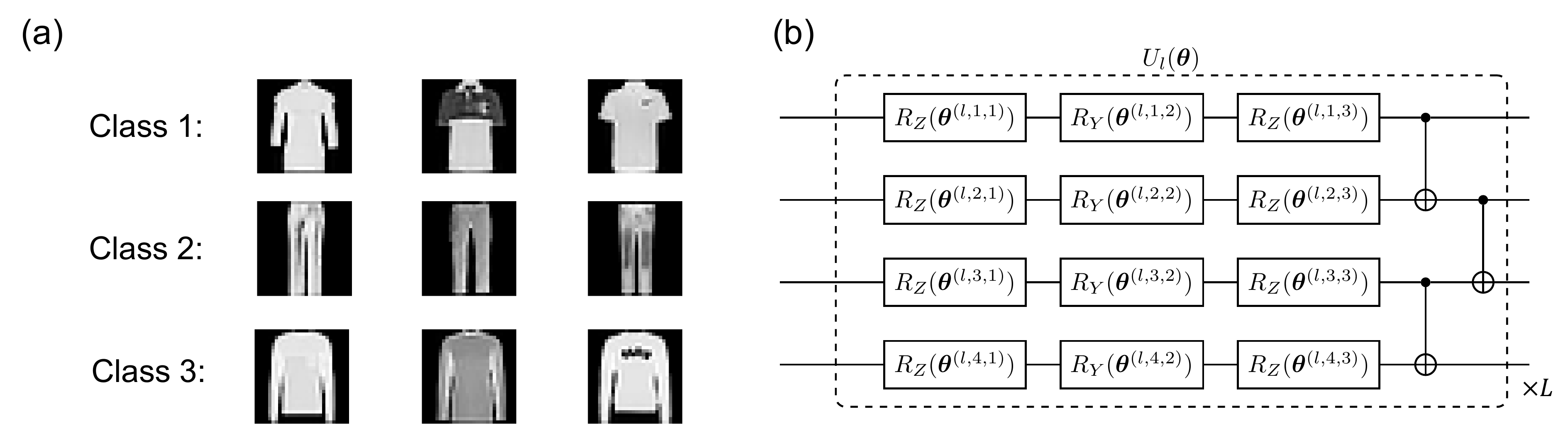}
	\caption{\small{\textbf{Visualization of image dataset and hardware-efficient ansatz.} (a) Image instances sampled from the Fashion-MNIST dataset. (b) The circuit architecture of the employed Hardware-efficient ansatz. The label `$\times L$' denotes the layer number, which means repeating the gates in the dashed box with $L$ times. }}
	\label{fig:append:dataset}
\end{figure*}
\textbf{Construction of QCs.}     
The quantum subroutine of QC consists of  the encoding circuit $U_E$ and the ansatz $U(\btheta)$. For all learning tasks, the hardware-efficient ansatz is employed  whose mathematical expression is $U(\btheta)=\prod_l^L U_l(\btheta)$. The layout of the hardware-efficient  ansatz follows the layer-wise structure and the gate arrangement at each layer is the same.  For $\forall l\in [L]$, $U_l(\btheta)=\bigotimes_{i=1}^N(\RZ(\btheta^{(l,i,1)})\RY(\btheta^{(l,i,2)})\RZ(\btheta^{(l,i,3)}))U_{ent}$ with $U_{ent}$ being the entanglement layer formed by $\CNOT$ gates. Fig.~\ref{fig:append:dataset}(b) depicts the adopted hardware-efficient ansatz  with $L$ layers.   

The encoding methods for the parity dataset classification and the digit images classification are different. The former uses the basis encoding method. Specifically, for a classical example $\bx\in \mathbb{R}^d$, the employed encoding unitary is $U_E(\bx)\ket{0}^{\otimes d}=\ket{\bx}$, which maps $\bx$ to a $2^d$ dimensional quantum state $U_E(\bx)\ket{0}^{\otimes d}$. The latter uses the amplitude encoding method. Given a normalized image $\bx\in \mathbb{R}^{64}$ with $\|\bx\|_2^2=1$, the corresponding unitary encodes it into a $6$-qubit state with $U_E(\bx)\ket{0}^{\otimes 6}=\sum_{j=1}^{64} \bx_j\ket{j}$.

The Pauli-based measure operators are used in learning Fashion-MNIST dataset. Since the preprocessed dataset contains $9$ classes, there are in total $9$ measure operators, i.e., $o^{(1)}=X\otimes X \otimes\mathbb{I}^{\otimes 8}$, $o^{(2)}=X\otimes Y \otimes\mathbb{I}^{\otimes 8}$, $o^{(3)}=X\otimes Z \otimes\mathbb{I}^{\otimes 8}$, $o^{(4)}=Y\otimes X \otimes\mathbb{I}^{\otimes 8}$, $o^{(5)}= Y\otimes Y \otimes\mathbb{I}^{\otimes 8}$, $o^{(6)}=Y\otimes Z \otimes\mathbb{I}^{\otimes 8}$, $o^{(7)}=Z\otimes X \otimes\mathbb{I}^{\otimes 8}$, $o^{(8)}=Z\otimes Y \otimes\mathbb{I}^{\otimes 8}$, $o^{(9)}=Z\otimes Z \otimes\mathbb{I}^{\otimes 8}$. 

\textbf{Multilayer Perceptron}. To better justify the capability and performance of QCs, we apply the multilayer perceptron (MLP) as the reference \cite{goodfellow2016deep}. MLP is composed of an input layer,  $L$ hidden layers with $L\geq 1$, and an output layer. The dimension of the input layer is equivalent to the feature dimension of the input. ReLU activations are added in the hidden layer to perform nonlinear transformation. In the output layer, the activation function, Softmax, is employed. The number of layers $L$ depends on the assigned tuples $\{n, N_t, T\}$. 

\textbf{Convolutional neural network}. In the task of image classification, convolutional neural networks (CNNs) is employed as the reference \cite{goodfellow2016deep}. The employed  CNN is formed by two convolutional layers and one fully-connected layer. ReLU activations and the pooling operation are added in the hidden layer to perform nonlinear transformation. The number of channels for the first convolutional layer is fixed to be $8$ and the corresponding kernel size is $9\times 9$. The kernel size of the pooling operation applied to the two convolutional layers is $2\times 2$. The kernel size for the second convolutional layer is fixed to be $5\times 5$ but the number of output channels is varied depending on the settings in Alg.~\ref{alg:dymanic-loss-curve}. For the sake of fair comparison, the number of output channels is set as $2, 6, 15, 30, 50, 75$, where the corresponding number of parameters is $860$, $1284$, $2238$, $3828$, $5948$, and $8598$, respectively.

\textbf{Optimizer and other hyper-parameters.} The adaptive gradient descent method, named AdaGrad optimizer \cite{duchi2011adaptive}, is used to optimize QCs and MLPs. Compared to the vanilla gradient descent method, AdaGrad permits better performance, since it adapts the learning rate for each feature depending on the estimated geometry of the problem. In the task of parity learning, the initial learning rate is set as $\eta=0.5$ for QC and $\eta=0.01$ for MLP, respectively. For both classifiers, the batch size is fixed to be $4$. In the task of image classification, the initial learning rate is set as $\eta=0.05$ for QC and $\eta=0.01$ for CNN, respectively. The batch size for both classifiers is set as $1$. To make a fair comparison, the hyper-parameter settings applied to QC and CC, especially for those relating to the computational resources, are required to keep to be the same. Specifically, in each comparison, the employed loss function, the train examples $n$, the number of trainable parameters  $N_t$, and the number of epochs $T$ applied to QC and CC should be identical. 

\textbf{Curve fitting method}. To capture the risk curve, Alg.~\ref{alg:dymanic-loss-curve} requests a curve fitting method.  For all experiments, we adopt the polynomial fitting to derive the risk curve by using the collected results. The least squares method in determining the best fitting functions.

\textbf{Source code.} The source code used in numerical simulations will be available at Github repository \url{https://github.com/yuxuan-du/Problem-dependent-power-of-QNNs}.

\begin{figure*}
	\centering
	\includegraphics[width=0.91\textwidth]{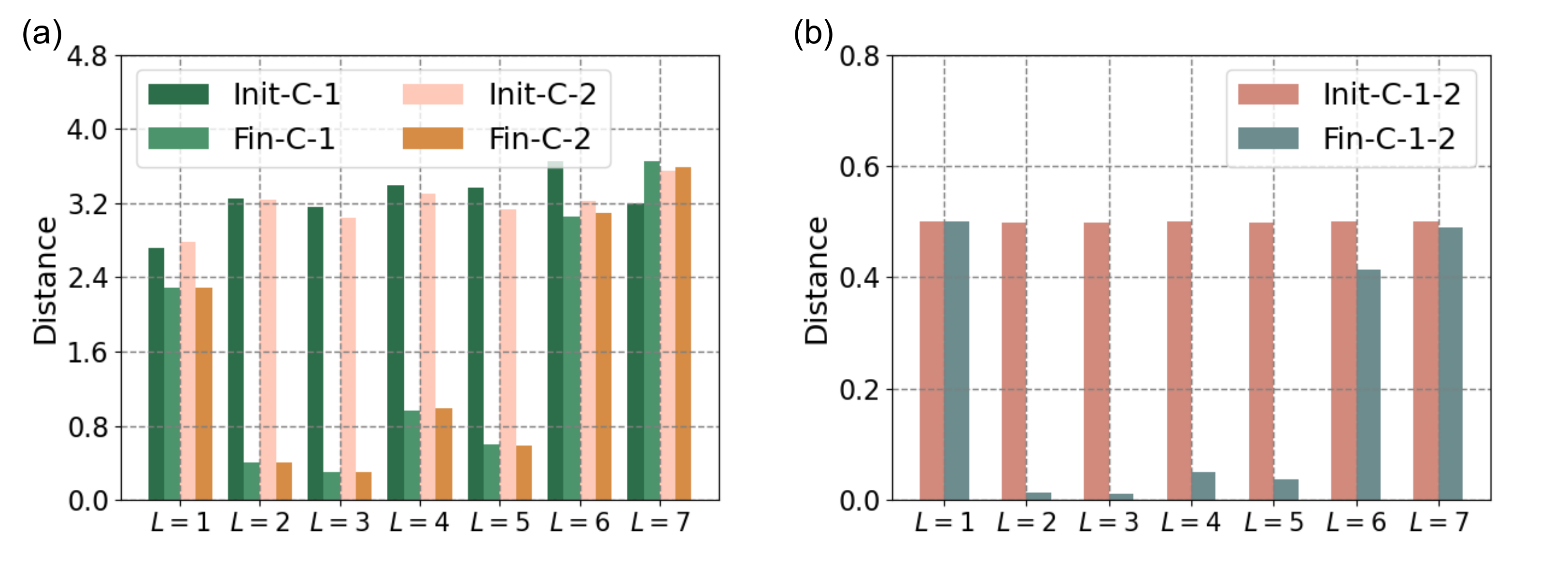} \caption{\small{\textbf{Geometric properties of the quantum feature states on parity dataset.} (a) The averaged performance of QC evaluated by $\mathcal{M}_1$ defined in Eq.~(\ref{eqn:append:metric-1}). The label `Init-C-$k$' with $k=1,2$ refers that the value of $\mathcal{M}_1^{(k)}$ at the initialization. Similarly, the label `Final-C-$k$' with $k=1,2$ refers that the value of $\mathcal{M}_1^{(k)}$ when the training of QC is completed. (b) The averaged performance of QC evaluated by $\mathcal{M}_2$ defined in Eq.~(\ref{eqn:append:metric-2}). The label `Init-C-$1$-$2$' (`Final-C-$1$-$2$') refers that the value of $\mathcal{M}_2$ before and after training of QC. The label `$L=a$' in the $x$-axis stands for that the layer number of hardware-efficient \revise{ansatz} is $a$.}}
	\label{fig:append:QC_parity-feature-states}
\end{figure*}

\subsection{Simulation results of the binary classification for the parity dataset}

\begin{figure*}
	\centering
	\includegraphics[width=0.99\textwidth]{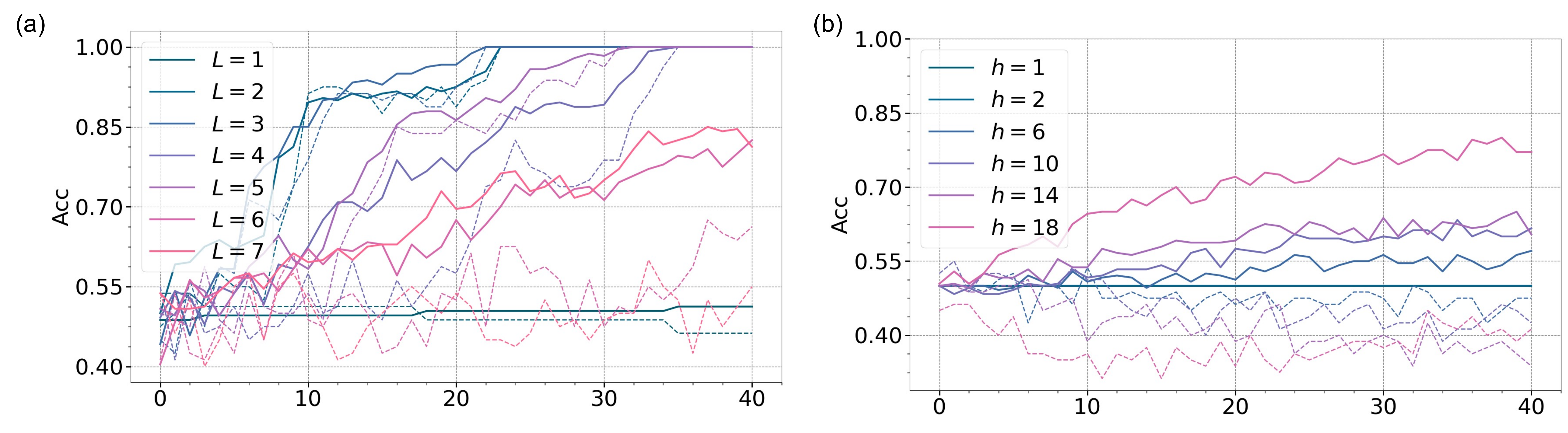}
\caption{\small{\textbf{Train (test) accuracy versus epoch on parity dataset.} (a) Train accuracy and test accuracy of QC with the varied layer number. The label `$L=a$' refers that the layer number used in hardware-efficient ansatz is $a$. The solid line and the dashed line separately correspond to the train and test accuracies of QC. (b) Train accuracy and test accuracy of MLP with the varied number of hidden neurons. The label `$h=a$' refers that the number of neurons is $a$. The solid and dashed lines have the same meaning with those in QC.}}
	\label{fig:append:QC_parity}
\end{figure*}

\textbf{The feature states before and after training.} We explore the geometric  properties of feature states when the layer number of hardware-efficient ansatz varies from $L=1$ to $L=7$. Other settings are identical to those introduced in the main text. Condition (i) in Lemma \ref{thm:condit-NC-QNN} is evaluated by the metric 
\begin{equation}\label{eqn:append:metric-1}
	\mathcal{M}_1^{(k)}=\sum_{i=1}^{n_c}\|\rho^{(i,k)}-\bar{\rho}^{(k)}\|,
\end{equation}
where the number of train examples $\{\rho^{(i,k)}\}_{i=1}^{n_c}$ belonging to the $k$-th class is $n_c$ and $\bar{\rho}^{(k)}$ refers to their class-feature mean. Since parity learning is a binary classification task, Condition (ii)  in Lemma \ref{thm:condit-NC-QNN} is evaluated by 
\begin{equation}\label{eqn:append:metric-2}
	\mathcal{M}_2=\Tr(\bar{\rho}^{(0)}\bar{\rho}^{(1)}). 
\end{equation}

The geometric properties of the feature states in the measure of $\mathcal{M}_1^{(k)}$ and $\mathcal{M}_2$ are visualized in Fig.~\ref{fig:append:QC_parity-feature-states}. The left panel shows that when $L\in\{2,3,4,5\}$, both the value of $\mathcal{M}_1^{(1)}$ (highlighted by the green color) and $\mathcal{M}_1^{(2)}$ (highlighted by the pink color) decrease  from $\sim 3.2$ (epoch $t=0$) to $\sim 0.5$ (epoch $t=40$).  These results comply with Condition (i) in the sense that the feature states in the same class concentrates to the class-feature mean and leads to the low empirical risk. By contrast, when $L$ is too small or too large, the value of $\mathcal{M}_1^{(1)}$ changes subtly before and after optimization, which is above $3.2$. The large deviation of feature states incurs the degrade performance of QC. The right panel depicts that when  $L\in\{2,3,4,5\}$, the value of $\mathcal{M}_1^{(2)}$  decreases from $0.5$ (epoch $t=0$) to $0.05$ (epoch $t=40$). This reduction means that the class-feature means are maximally separated and thus ensure a good learning performance. On the contrary, when $L\in\{1, 6, 7\}$, the the value of $\mathcal{M}_1^{(2)}$ oscillates around $0.5$, which implies that the class-feature means $\bar{\rho}^{(1)}$ and $\bar{\rho}^{(2)}$ are highly overlapped. 
 
\textbf{The learning dynamics of QC and MLP}. Fig.~\ref{fig:append:QC_parity} visualizes the learning dynamics of QC and MLP with respect to the varied trainable parameters. The left panel indicates that when the layer number is $L=2, 3, 4$, both train and test accuracies of QC fast converge to $100\%$ with $25$ epochs. When $L=1$, both train and test accuracies oscillate to $50\%$. When $L=7$, the number of train data becomes insufficient and the overfitting phenomenon appears. These results accord with the U-shaped risk curve of QCs. The right panel shows that when the number of hidden neurons ranges from $h=1$ to $h=18$, the test accuracy of MLP is no higher that $55\%$. These results reflect the incapability of MLP in learning parity dataset compared with QCs.

\subsection{Simulation results of multi-class classification for the Fashion-MNIST images dataset} 

\textbf{The feature states before and after training.} Here we discuss the geometric properties of feature states when the layer number of hardware-efficient ansatz  varies from $L=2$ to $L=150$. The metrics 
	$\mathcal{M}_1^{(k)}$ and $\mathcal{M}_2$ defined in Eqs.~(\ref{eqn:append:metric-1}) and (\ref{eqn:append:metric-2}) are employed. In the measure of $\mathcal{M}_2$, since the performance of QC for any two classes is similar, we only study the first two classes for ease of visualization.

Fig.~\ref{fig:append:QC_FMNIST-feature-states} depicts the geometric properties of the feature states in the measure of $\mathcal{M}_1^{(k)}$ and $\mathcal{M}_2$. The left panel shows that for all settings with $L\in\{2,5,25, 50,100,150\}$, the   value $\mathcal{M}_1^{(k)}$  at the initial step and the final step is very similar and $\mathcal{M}_1^{(k)}$ is larger than $0.2$ for $\forall k\in\{1,2,...,9\}$. These results indicate that QC cannot satisfy Condition (i) when learning Fashion-MNIST dataset, where the feature states from the same class cannot collapse to a unique point. Moreover, when we examine the performance of intra-class, the right panel implies that after training, the class-feature means of QC are still highly overlapping. The distance for all settings of $L$ is above $0.3$. The inability to achieve the optimal training loss shows the the limited power of QC on learning Fashion-MNIST dataset.
 
 \begin{figure*}
	\centering
	\includegraphics[width=0.99\textwidth]{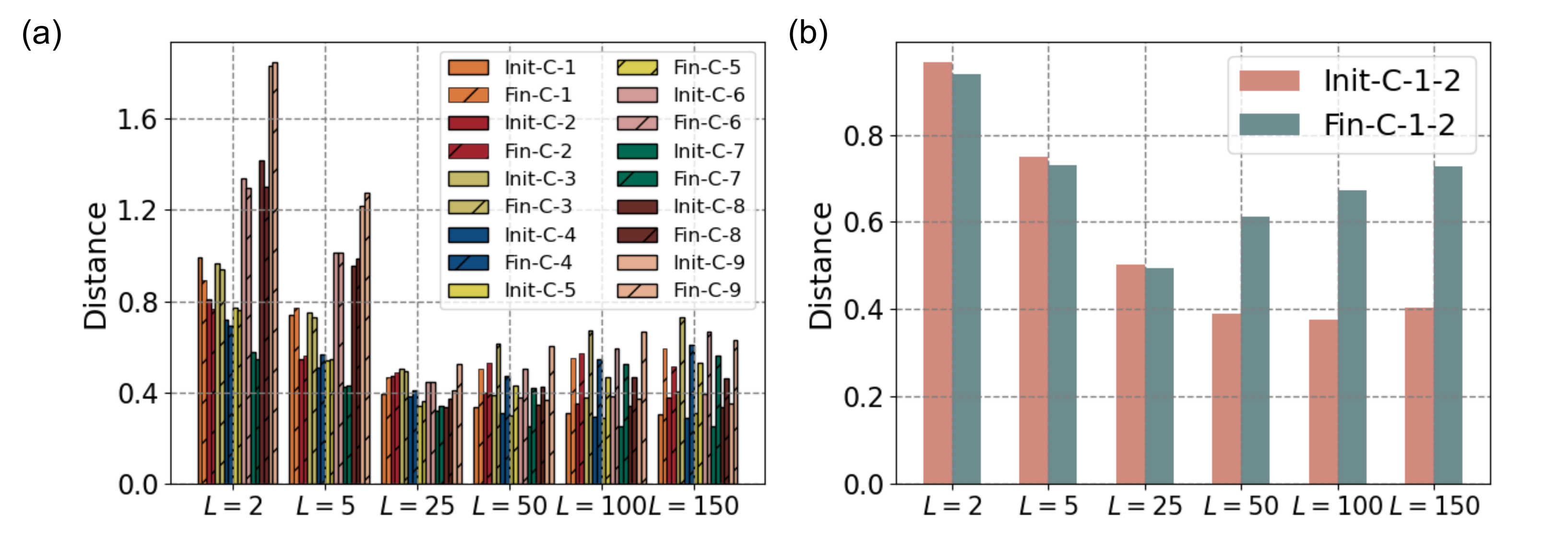} \caption{\small{\textbf{Geometric properties of the quantum feature states on Fashion-MNIST dataset.} (a) The averaged performance of QC evaluated by $\mathcal{M}_1$ defined in Eq.~(\ref{eqn:append:metric-1}).  (b) The averaged performance of QC evaluated by $\mathcal{M}_2$ defined in Eq.~(\ref{eqn:append:metric-2}). All labels have the same meaning with those introduced in Fig.~\ref{fig:append:QC_parity-feature-states}.}}
	\label{fig:append:QC_FMNIST-feature-states}
\end{figure*}
 
\begin{figure*}
	\centering
	\includegraphics[width=0.99\textwidth]{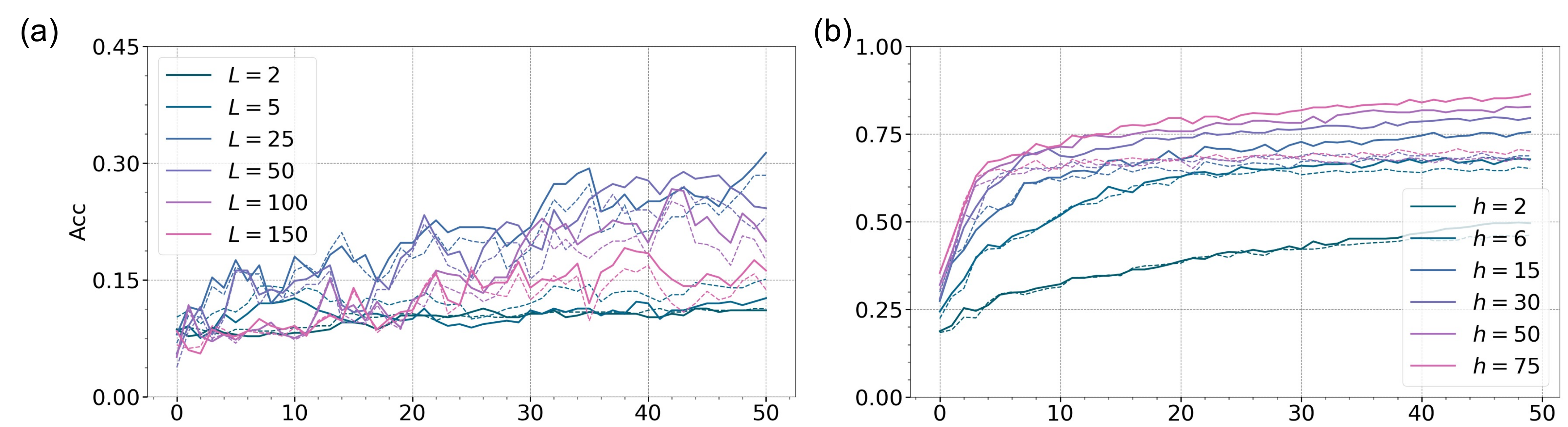}
\caption{\small{\textbf{Train (test) accuracy versus epoch on Fashion-MNIST dataset.} (a) Train accuracy and test accuracy of QC with the varied layer number. The labels have the same meaning with those presented in Fig.~\ref{fig:append:QC_parity}. (b) Train accuracy and test accuracy of CNN with the varied number of trainable parameters. The label `$h=a$' refers that the number of output channels at the second layer is $a$. The solid and dashed lines have the same meaning with those in QC.}}
	\label{fig:append:QC_Fashion}
\end{figure*}

\textbf{The learning dynamics of QC and CNN}. Fig.~\ref{fig:append:QC_Fashion} depicts the learning dynamics of QC and CNN with the varied number of trainable parameters. The left panel indicates that QC achieves the best performance when the layer number is $L\in[25, 100]$, where the corresponding number of parameters ranges from $750$ to $3000$. In these settings, both train and test accuracies of QC are around $30\%$ after $50$ epochs. When $L<25$ or $L>100$, both train and test accuracies oscillate at $15\%$.  These results accord with the U-shaped risk curve of QCs. The right panel shows that the train and test accuracies of CNN are steadily growing with the increased number of channels. That is, when the number of channels at the second layer is not less than $6$, both the train and test accuracies are higher than $60\%$.  These results indicate that the employed QC does not have potential advantages in learning image dataset compared with CNN.

\end{document}